\documentclass[english]{scrartcl}
\usepackage[T1]{fontenc}
\usepackage{lmodern}
\usepackage[utf8]{inputenc}
\usepackage[english]{babel}
\usepackage{amsmath}
\usepackage{amsthm}
\usepackage{amssymb}
\usepackage{prettyref}
\usepackage{mathtools}
\usepackage{enumitem}
\usepackage[a4paper, margin=1in]{geometry}
\usepackage{fancyhdr}
\usepackage{graphicx}
\usepackage[bookmarks]{hyperref}

\makeatletter

\numberwithin{equation}{section}

\theoremstyle{plain}
\newtheorem{thm}{\protect\theoremname}[section]
\theoremstyle{definition}
\newtheorem{defn}[thm]{\protect\definitionname}
\theoremstyle{remark}
\newtheorem{rem}[thm]{\protect\remarkname}
\theoremstyle{plain}
\newtheorem{lem}[thm]{\protect\lemmaname}
\theoremstyle{plain}
\newtheorem{prop}[thm]{\protect\propositionname}
\theoremstyle{plain}
\newtheorem{cor}[thm]{\protect\corollaryname}
\theoremstyle{remark}
\newtheorem*{claim*}{\protect\claimname}
\@ifundefined{date}{}{\date{}}

\newrefformat{prop}{Proposition~\ref{#1}}

\newrefformat{cor}{Corollary~\ref{#1}}

\newrefformat{rem}{Remark~\ref{#1}}
\newrefformat{ssec}{Subsection~\ref{#1}}

\makeatother

\providecommand{\theoremname}{Theorem}
\providecommand{\definitionname}{Definition}
\providecommand{\lemmaname}{Lemma}
\providecommand{\propositionname}{Proposition}
\providecommand{\remarkname}{Remark}
\providecommand{\claimname}{Claim}

\providecommand{\claimname}{Claim}
\providecommand{\corollaryname}{Corollary}
\providecommand{\definitionname}{Definition}
\providecommand{\lemmaname}{Lemma}
\providecommand{\propositionname}{Proposition}
\providecommand{\remarkname}{Remark}
\providecommand{\theoremname}{Theorem}

\begin{document}
\global\long\def\tr{\textnormal{tr}}%
\global\long\def\d{\textnormal{d}}%
\global\long\def\R{\mathbb{R}}%
\global\long\def\E{\mathbb{E}}%
\global\long\def\P{\mathbb{P}}%
\global\long\def\V{\mathbb{V}}%
\global\long\def\Z{\mathbb{Z}^{3}}%
\global\long\def\id{\textnormal{id}}%
\global\long\def\phiak{(h_{y})_{\alpha}(k)}%
\global\long\def\Elin{\mathfrak{E}^{\textnormal{lin}}}%
\global\long\def\Ebos{\mathfrak{E}^{\textnormal{B}}}%
\global\long\def\supp{\textnormal{supp}}%
 
\global\long\def\rangleka{\rangle_{\Gamma}}%
\global\long\def\sumak{\sum_{\alpha\in\mathcal{I}_{k}}}%
\global\long\def\sumal{\sum_{\alpha\in\mathcal{I}_{l}}}%
\global\long\def\kinG{k\in\Gamma}%
\global\long\def\linG{l\in\Gamma}%
\global\long\def\Heff{\mathbb{H}^{\textnormal{eff}}}%
\global\long\def\xo{\frac{\pi}{2}}%
\global\long\def\mr#1{\mathcal{R}^{#1}}%

\global\long\def\cak{c_{\alpha}(k)}%
\global\long\def\ccak{c_{\alpha}^{\ast}(k)}%
\global\long\def\fock{\mathcal{F}}%
 
\global\long\def\domH{L^{2}(\Lambda,\d y)\otimes\mathcal{H}_{N}^{-}}%
\global\long\def\BF{B_{\textnormal{F}}}%
\global\long\def\BFc{\BF^{c}}%
\global\long\def\kF{k_{\textnormal{F}}}%
\global\long\def\cgl{c_{\gamma}(l)}%
\global\long\def\ccgl{c_{\gamma}^{\ast}(l)}%

\title{Effective Polaron Dynamics of an Impurity Particle Interacting with a Fermi Gas}

\author{Duc Viet Hoang\thanks{Fachbereich Mathematik, University of Tübingen, Auf der Morgenstelle
10, 72076 Tübingen, Germany. E-mail: \texttt{viet.hoang@uni-tuebingen.de}}\and Peter Pickl \thanks{Fachbereich Mathematik, University of Tübingen, Auf der Morgenstelle
10, 72076 Tübingen, Germany. E-mail: \texttt{p.pickl@uni-tuebingen.de}}}

\frenchspacing

\date{\today}
\maketitle
\frenchspacing

\begin{abstract}
We study the quantum dynamics of a homogeneous ideal Fermi gas coupled
to an impurity particle on a three-dimensional box with periodic boundary
condition. For large Fermi momentum $\kF$, we prove that the effective
dynamics is generated by a Fröhlich-type polaron Hamiltonian, which
linearly couples the impurity particle to an almost-bosonic excitation
field. Moreover, we prove that the effective dynamics can be approximated
by an explicit coupled coherent state. Our method is applicable to a range of
interaction couplings, in particular including interaction couplings of 
order 1 and time scales of the order $\kF^{-1}$.
\end{abstract}

\begin{abstract}
\tableofcontents{}
\end{abstract}

\section{Introduction}

The study of impurities in quantum gases has garnered considerable
attention due to its relevance in various physical contexts, ranging
from solid-state physics to cold atom experiments. In this context
quasi-particles such as polarons stand as intriguing entities emerging
from the interaction of a single impurity particle with a surrounding
medium. The concept of a polaron, originally introduced by Lev Landau
to study the motion of an electron in a dielectric crystal \cite{1965}, most famously
emerges from the celebrated Fröhlich Hamiltonian in second quantization
formalism describing electron-phonon interactions  \cite{Froehlich1954}. Subsequently, the
polaron concept was extended to all kind of surrounding media including
Bose and Fermi gases. The formation conditions and properties of polarons
are believed to play a central role to understand the transport properties
and the effective mass of impurities within the host material. 

In this article, we study with mathematical rigor the dynamics of
an impurity particle immersed in a dense gas of fermions as surrounding
medium. Interactions between fermions are neglected and we assume
that the impurity particle interacts with the Fermi gas via a finite-range 
potential in momentum space.
The initial state $\psi$ of the system is a product state between
the impurity state and a filled Fermi ball. This mathematical framework
finds resonance with recent experimental and theoretical advancements
in the study of ultracold atoms \cite{Schirotzek2009,Knap2012,Cetina2015}.
We emphasize that our assumption on the interaction potential aligns with the growing interest 
among experimentalists and theorists in cases where the impurity is charged, 
and hence the interaction potential is long-ranged in position space \cite{Christensen2022,Mysliwy2024}.
We show that the effective dynamics of the system is governed by a
Fröhlich-type Hamiltonian, which linearly couples the impurity particle
to an almost-bosonic excitation field. More specifically, the excitations
relative to the filled Fermi ball are up to a constant described by
the Hamiltonian 

\begin{equation}
\mathbb{H}^{\text{F}}=(-\Delta_{y})\otimes1+1\otimes\mathbb{D}_{\text{B}}+\Phi(h_{y})\label{eq:HF}
\end{equation}
with $(-\Delta_{y})$ describing the kinetic energy of the impurity
particle, $\mathbb{D}_{\text{B}}$ describing the kinetic energy of
the excitation field and $\Phi(h_{y})\coloneqq c^{\ast}(h_{y})+c(h_{y})$
the linear coupling between impurity particle and excitations. The
operators $c^{\ast}$ and $c$ describing this excitation field coincide
with those introduced in a series of pioneering studies on the correlation
energy of interacting fermions \cite{Benedikter2019,Benedikter2021,Benedikter2021a}.
We note that an effective Hamiltonian of a similar type to \eqref{eq:HF}
has recently been derived in another microscopic setting involving
a tracer particle interacting with excitations of a Bose Einstein
condensate \cite{Lampart2022,Mysliwy2020}.\\
Subsequently, we show that the effective time evolved state can be
up to a phase factor approximated by a time-dependent coupled coherent
state $W(\eta_{t})\phi\otimes\Omega$ where $W$ is the Weyl operator of the excitation field
which is simply parameterized by a function $\eta_{t}$ and $\Omega$ is the Fock space vacuum. 
An explicit expression for $\eta_{t}$ is derived which allows for determining
the number of collective excitations over time. We believe that such
quantities are particularly helpful to gain deeper insights into
the formation process of quasi-particles as studied in experiments
such as \cite{Cetina2016}. Eventually, we show that the linear coupling
term $\Phi(h_{y})$ cannot be omitted in the effective description
but adds a leading order effect to the effective dynamics in the setting of
interaction couplings of order 1.

Our results hold for a variety of time scales and couplings, describing
different interaction strengths and mass ratios, which will be specified
in the subsequent section. We note that the same microscopic model
has been studied in \cite{Jeblick2017,Jeblick2018,mitrouskas2021effective}
but with very specific choices of couplings different from ours, leading
to an effective decoupling of impurity and gas.

\subsection{The microscopic model}

We consider an impurity particle interacting with $N$ spinless fermions
on a 3-dimensional box with periodic boundaries described by $\Lambda\coloneqq\mathbb{T}^{3}\coloneqq\R^{3}/(2\pi\mathbb{Z}^{3})$.
The system is described by a state in the Hilbert space $\domH$
with $\mathcal{H}_{N}^{-}=L^{2}(\Lambda)^{\wedge N}$ where $y$ is
the coordinate of the impurity particle and $\{x_{i}\}_{i=1,\ldots,N}$
are the coordinates of the fermions. The Hamiltonian for our main model of the system
is given by 
\begin{equation}
H_{N}=-\beta\Delta_{y}+\sum_{i=1}^{N}(-\Delta_{x_{i}})+\lambda\sum_{i=1}^{N}V(x_{i}-y)\label{eq:microscopic-Hamiltonian}
\end{equation}
and parameterized by  $\beta,\lambda>0$.
Note that the different parts of the Hamiltonian on different tensor 
components of our Hilbert space writing, i.e. we used the short-hand notation writing, for example,
$-\Delta_{y}\coloneqq -\Delta_{y}\otimes1$ for the Laplacian acting on the impurity particle. 
The interaction $V$ is assumed to have a Fourier transform
$\hat{V}$ with compact support satisfying $\hat{V}(k)=\hat{V}(-k)$ and $\hat{V}(k)\geq 0$
for all $k\in\Z$. It is well-known that under this assumption the
Hamiltonian \eqref{eq:microscopic-Hamiltonian} defines a self-adjoint
operator which generates by Stone's theorem the unitary time evolution
$e^{-H_{N}t}$.

We are interested in the dynamics of the system governed by the time-dependent
Schrödinger equation of the form 
\begin{equation}
i \hbar \frac{\d}{\d t}\psi_{t}=H_{N}\psi_{t},\qquad\psi_{0}\in\domH.\label{eq:Schroedinger-eq}
\end{equation}
where $\hbar$ is the reduced Planck constant. 
For mathematical convenience we will always set $\hbar=1$ unless explicitly stated otherwise.

Note that the filled Fermi ball is a ground state for the non-interacting Fermion system. It is
 non-degenerate and explicitly given by 
\begin{equation}
\Omega_{0}\coloneqq\bigwedge_{k\in\BF}f_{k},\qquad f_{k}(x)\coloneqq\frac{\exp(ikx)}{(2\pi)^{3/2}}\in L^{2}(\Lambda).\label{eq:}
\end{equation}
We choose the initial state to be of product form
\begin{equation}
\psi_{0}\big(y;x_{1},\ldots x_{N}\big)\coloneqq\phi(y)\otimes\Omega_{0}(x_{1},\ldots x_{N}),\label{eq:initial-state}
\end{equation}
with a general state $\phi$ for the impurity particle
i.e., the system is initially prepared in a state describing a ground
state of the ideal Fermi gas which does not interact with the impurity
particle.

Furthermore, we choose the Fermi momentum $k_{F}$ to be our parameter
of the system in the sense that the particle number is defined as 

\begin{equation}
N\equiv N(\kF)\coloneqq|\BF|,\qquad\BF\coloneqq\{k\in\mathbb{Z}^{3}:|k|\leq\kF\},
\end{equation}
i.e. the particle number $N$ of the Fermi gas is chosen such that
the Fermi ball is completely filled. Note that the average density
is in this case proportional to the number $N$ of gas particles due
to the following relation
\begin{equation}
\kF=\left(\frac{3}{4\pi}\right)^{1/3}N^{\frac{1}{3}}+\mathcal{O}(1)
\end{equation}
which is a consequence of Gauss' counting argument.

\subsection{Relevant parameters and time scales} \label{ssec:couplings}

In the following, we present the ranges of $\beta$ and $\lambda$
we are aiming for, and discuss the physical meaning of the parameters. 
In addition, it is important to discuss on which time scales our results
hold. 
\begin{itemize}
\item The parameter $\beta$ determines the mass ratio of the impurity and Fermi gas particles
such that  $\beta \ll 1$ corresponds to a relatively heavy, $\beta \gg 1$ to a relatively light impurity particle.
The case of equal masses corresponds to $\beta=1$. Our work requires the
restriction $\beta \in o(\kF) $ see \prettyref{rem:coherent-bound}. Since $\kF$ tends to infinity,
this allows to include all three cases.
\item The parameter $ \lambda \in [\kF^{-1/6}, 1 ] $ models the coupling strength between the Fermi
gas and the impurity. For small $\lambda$ we expect a decoupling
between the gas and the impurity in the sense that the time evolution
given by \eqref{eq:Schroedinger-eq} does not entangle an initial state
of product form $\phi\otimes\Omega_{0}$. Such a result was shown
in \cite{mitrouskas2021effective} with $\beta=1$ and $\lambda=\kF^{-1/2}$
in three dimensions and \cite{Jeblick2017,Jeblick2018} with $\beta=\lambda=1$
in two dimensions. Note that our interaction coupling $\lambda$ is much larger than in \cite{mitrouskas2021effective} and in particular, we are able to include $\lambda=1$.
As will be discussed in \prettyref{rem:Fermi-blockade}, it is necessary to include excitations of the Fermi gas into the
effective description for $\lambda=1$.
\item Our approximations will apply for times $t \in \mathcal{O} (\kF^{-1} \lambda^{-1})$, so that a weaker coupling strength $\lambda$ corresponds to slightly larger times.
The times $t\in\mathcal{O}(\kF^{-1})$ are on the time scale of the fermions near the Fermi surface, which 
have an approximate momentum of $\kF$ and thus travel a distance of order 1 in this time. 
We are therefore able to enter a time scale where the impurity
particle can resolve the motion of the fermions.
We remark that the results in \cite{mitrouskas2021effective} can
be transferred to this setting of $\lambda=1$, however, allowing only for
shorter time scales of $t \in o (\kF^{-1})$.
\end{itemize}

We remark that results obtained in the previously mentioned parameter range remain valid under re-scaling by multiplying the Hamiltonian by an overall factor and absorbing it by re-scaling the time variable.
Although we will stick to the above choices, it might still be helpful for the reader to see some connections to other scaling regimes by re-scaling.
We present two of them briefly:
\begin{itemize}
  \item Time scales of order 1 are obtained by multiplying the Hamiltonian by an overall factor of $\kF^{-1}$. In this case, the factor $\kF^{-1}$ appears
  as new parameters in front of the kinetic energy of the Fermi gas corresponding to a heavy fermion regime.
  \item In recent years the so-called semiclassical regime has been
  widely studied in the analysis of dense Fermi gases, as can be seen for example in \cite{Benedikter2022,Saffirio2023}.
  This regime is associated with identifying $\hbar \coloneqq \kF^{-1}$ as small parameter instead of setting $\hbar=1$. 
  It is achieved by multiplying the Hamiltonian by an overall factor of $\kF^{-2}$ and absorbing $\kF^{-1}$ in the time variable. 
  For this new time scale, our theorem makes a statetment for times of order 1. 
  The re-scaled Schrödinger equation takes the form of 
  \begin{equation}
  i\hbar\partial_{t}\psi_{t}=\left(\hbar^{2}\beta (-\Delta_{y})+\hbar^{2}\sum_{i=1}^{N}(-\Delta_{x_{i}})+ \lambda\kF^{-2} \sum_{i=1}^{N}V(x_{i}-y)\right)\psi_{t}
  \end{equation}
  One identifies $\lambda\kF^{-2}\in [\kF^{-13/6},\kF^{-2}]$ as re-scaled interaction coupling.
\end{itemize}

\section{Preliminaries}

\paragraph*{Second quantization}

It is convenient to consider $\domH$ as the $N$-particle sector of $L^{2}(\Lambda)\otimes\fock$
with the fermionic Fock space $\mathcal{F}$ constructed over $L^{2}(\Lambda)$.
This way, we have access to the powerful formalism of second quantization
with the fermionic creation operator $a_{p}^{\ast}$ creating a particle
with momentum $p\in\mathbb{Z}^{3}$ and the annihilation operator
$a_{p}$ annihilating a particle with momentum $p\in\mathbb{Z}^{3}$.
Those operators satisfy the canonical anticommutation relations (CAR)
\begin{equation}
\forall p,q\in\mathbb{Z}^{3}:\{a_{p},a_{q}^{\ast}\}=\delta_{p,q},\qquad\{a_{p},a_{q}\}=0=\{a_{p}^{\ast},a_{q}^{\ast}\}.\label{eq:CAR}
\end{equation}
Furthermore we introduce the fermionic number operator $\mathcal{N}\coloneqq\sum_{p\in\mathbb{Z}^{3}}a_{p}^{\ast}a_{p}$
and the vacuum $\Omega$ satisfying $a_{p}\Omega=0$ for all $p\in\mathbb{Z}^{3}$.

We lift our $N$-particle Hamiltonian $H_{N}$ to Fock space as

\begin{equation}
\mathbb{H}=\underbrace{-\beta\Delta_{y}}_{\eqqcolon h_{0}}+\underbrace{\sum_{k\in\Z}|k|^{2}a_{k}^{\ast}a_{k}}_{\eqqcolon\mathbb{H}^{\textnormal{kin}}}+\underbrace{\lambda\sum_{k,p\in\Z}\hat{V}(k)e^{iky}a_{p}^{\ast}a_{p-k}}_{\eqqcolon\mathbb{V}}
\end{equation}
which agrees with $H_{N}$ if restricted to $\domH$. We denote by $\langle\cdot,\cdot\rangle$
the inner product on $\domH$ and by $\|\cdot\|$ the induced norm
if not stated otherwise.

We will mostly use the abuse of notation $\mathbb{A}\equiv1\otimes\mathbb{A}$
as operator on $L^{2}(\Lambda)\otimes\fock$ where $\mathbb{A}$ acts
as an operator on the Fock space part.

\paragraph*{Particle-hole transformation}

In our analysis, the primary objective is to focus on excitations relative to the non-interacting Fermi ball. In particular, we want to use a description of our fermionic system in which the non-interacting Fermi ball $\Omega_0 = \prod_{k\in \BF} a^{\ast}_k \Omega $ is mapped to the vacuum.
To achieve this, we employ the particle-hole transformation, which is a specific type of fermionic Bogoliubov transformation as creation operators are mapped to linear combinations of creation and annihilation operators while preserving the CAR.
The \emph{particle-hole transformation} is defined as the map $R:\fock\to\fock$
satisfying 
\begin{equation}
R^{\ast}a_{k}^{\ast}R\coloneqq\begin{cases}
a_{k}^{\ast} &  \text{if } k\in\BFc,\\
a_{k} & \text{if } k\in\BF
\end{cases}\quad \text{and} \quad R\Omega\coloneqq \Omega_{0}\ .\label{eq:particle-hole-trafo}
\end{equation}
It is easy to check that the map is well-defined, unitary and satisfies
$R^{-1}=R^{\ast}=R$. 

With this, we can re-write the initial
state \eqref{eq:initial-state} representing a non-interacting impurity
particle and a Fermi gas as
\begin{equation}
\psi_{0}=\phi\otimes\Omega_{0}=(1\otimes R)\left(\phi\otimes\Omega\right)\eqqcolon R\psi.
\end{equation}
Later on, we will mostly use the product state $\psi=\phi\otimes\Omega$
of the impurity and the vacuum instead of $\psi_{0}$. \\
Furthermore, we define
\begin{equation}
E_{N}^{\text{pw}}\coloneqq\sum_{k\in\BF}|k|^{2}=\langle R\Omega,H_{N},R\Omega\rangle\label{eq:E-pw}
\end{equation}
to be the energy of the non-interacting Fermi ball.

Of greatest interest is of course the action of the particle-hole
transformation on the microscopic Hamiltonian as generator of the
dynamics. The conjugation with $R$ of $\mathbb{H}=-\beta\Delta_{y}+\mathbb{H}^{\textnormal{kin}}+\mathbb{V}$
yields

\begin{align}
\mathbb{H}_{0} & \coloneqq R^{\ast}\mathbb{H}^{\textnormal{kin}}R-E_{N}^{\text{pw}}=\sum_{k\in\Z}|k|^{2}R^{\ast}a_{k}^{\ast}RR^{\ast}a_{k}R-E_{N}^{\text{pw}}\\
 & =\sum_{k\in\BF}|k|^{2}a_{k}a_{k}^{\ast}+\sum_{k\in\BFc}|k|^{2}a_{k}^{\ast}a_{k}-E_{N}^{\text{pw}}\\
 & =\sum_{k\in\Z}e(k)a_{k}^{\ast}a_{k}\qquad\text{with }e(k)\coloneqq\begin{cases}
|k|^{2} & \text{if } k\in\BFc,\\
-|k|^{2} & \text{if } k\in\BF.
\end{cases}\label{eq:H_0-def}
\end{align}

Similarly, we can see that 
\begin{subequations}
\begin{align}
R^{\ast}\mathbb{V}R & =\lambda\sum_{k,p\in\Z}\hat{V}(k)e^{iky}R^{\ast}a_{p}^{\ast}RR^{\ast}a_{p-k}R\nonumber \\
 & =\lambda\sum_{k\in\Z}\sum_{\substack{p-k\in\BFc,\\
p\in\BF
}
}\hat{V}(k)e^{iky}a_{p}a_{p-k}+\lambda\sum_{k\in\Z}\sum_{\substack{p\in\BFc,\\
p-k\in\BF
}
}\hat{V}(k)e^{iky}a_{p}^{\ast}a_{p-k}^{\ast}\label{eq:bosonizable-V}\\
 & \quad+\lambda\sum_{k\in\Z}\sum_{\substack{p\in\BF,\\
p-k\in\BF
}
}\hat{V}(k)e^{iky}a_{p}a_{p-k}^{\ast}+\lambda\sum_{k\in\Z}\sum_{\substack{p\in\BFc,\\
p-k\in\BFc
}
}\hat{V}(k)e^{iky}a_{p}^{\ast}a_{p-k}.\label{eq:non-bosonizable-V}
\end{align}
\end{subequations}

For later purposes we shall introduce for $\varphi\in l^{2}(\mathbb{Z}^{3})$
the short-notation 
\begin{align}
b(\varphi) & \coloneqq\sum_{k\in\Z}\overline{\varphi(k)}\sum_{\substack{p\in\BFc,\\
p-k\in\BF
}
}a_{p-k}a_{p},\\
b^{\ast}(\varphi) & =\sum_{k\in\Z}\varphi(k)\sum_{\substack{p\in\BFc,\\
p-k\in\BF
}
}a_{p}^{\ast}a_{p-k}^{\ast}.\label{eq:short-hand1}
\end{align}
We can then write
\begin{equation}
R^{\ast}\mathbb{H}R=-\beta\Delta_{y}+\mathbb{H}_{0}+b^{\ast}(\tilde{h}_y)+b(\tilde{h}_y)+\mathcal{E}\label{eq:ph-transformed-H}
\end{equation}
with $\tilde{h}_{y}(k)\coloneqq\lambda\hat{V}(k)e^{iky}$
and $\mathcal{E}$ is given by the terms of \prettyref{eq:non-bosonizable-V}
since

\begin{align}
\sum_{k\in\Z}\sum_{\substack{p-k\in\BFc,\\
p\in\BF
}
}\hat{V}(k)e^{iky}a_{p}a_{p-k}= & \sum_{k\in\Z}\sum_{\substack{\tilde{p}\in\BFc,\\
\tilde{p}+k\in\BF
}
}\hat{V}(k)e^{iky}a_{\tilde{p}+k}a_{\tilde{p}}\nonumber \\
= & \sum_{k\in\Z}\sum_{\substack{p\in\BFc,\\
p-k\in\BF
}
}\hat{V}(k)e^{-iky}a_{p-k}a_{p}
\end{align}
where we used that $\hat{V}(-k)=\hat{V}(k)$.

\paragraph*{Almost-bosonic operators and patch decomposition }

Our effective description of the microscopic system described by \eqref{eq:microscopic-Hamiltonian}
will involve the emergence of almost-bosonic particles describing
pair excitations of the Fermi ball. Those pair excitations will be
delocalized over the Fermi surface in the sense that they correspond
to a linear combination of pairs of fermionic operators. As mentioned
before  the almost-bosonic pair operators which occur in this article
coincide with the ones introduced in the series of seminal works \cite{Benedikter2019,Benedikter2021,Benedikter2021a,Benedikter2023}
on the correlation energy of a weakly interacting Fermi gas. We give
a brief introduction to the construction of those operators with the
most relevant properties in this subsection and in \prettyref{sec:Useful-bosonization}. 

A key ingredient for the approximation of the microscopic fermionic
system by almost-bosonic excitations is the decomposition of the Fermi
surface into patches. This will allow to approximate the fermionic
kinetic energy term by a term quadratic in the almost-bosonic pair
operators.

Introduce the bisecting subset of $\Z\cap\supp\hat{V}$
\begin{equation}
\Gamma\coloneqq\left\{ (k_{1},k_{2},k_{3})\in\Z\cap\supp\hat{V}:(k_{3}>0\vee k_{3}=0), (k_{2}>0\vee k_{2}=k_{3}=0) ,k_{1}>0\right\} 
\end{equation}
allowing the decomposition $\Gamma\cup (-\Gamma)=\Z\cap\supp\hat{V}$.

The construction works as follows: 
\begin{enumerate}
\item Choose the number $M$ of patches satisfying 
\begin{equation}
N^{2\delta}\ll M\ll N^{\frac{2}{3}-2\delta},\qquad\delta\in(0,\frac{1}{6}).\label{eq:condition-on-M}
\end{equation}
The lower bound on $M$ is needed to control the number of momenta inside each
patch whereas the upper bound is needed to suppress Pauli's principle.
The choice of $\delta$ and $c$ will be taken later. 
\item Define equal-area disjoint patches $p_{\alpha}$ 
as follows 
\begin{itemize}
\item $p_{1}$ is spherical cap of area $4\pi/M$, 
\item decompose remaining semi-sphere into $\sqrt{M}/2$ collars, 
\item leave corridors of width $2R\coloneqq2\ \text{supp}\hat{V}$ between
adjacent patches,
\item define patches of southern semi-sphere by reflection $k\mapsto-k$.
\end{itemize}
Define $\omega_{\alpha}\in S^{2}$ as centers of the patch $p_{\alpha}$. A useful graphical sketch of this patch construction is given in Figure 1 of \cite{Benedikter2019}.
\item For given $\kinG$ define the set of north and south patch indices
\begin{align}
\mathcal{I}_{k}^{+} & \coloneqq\{\alpha\in\{1,\ldots,M\}\mid k\cdot\hat{\omega}_{\alpha}\geq N^{-\delta}\},\\
\mathcal{I}_{k}^{-} & \coloneqq\{\alpha\in\{1,\ldots,M\}\mid k\cdot\hat{\omega}_{\alpha}\leq-N^{-\delta}\}
\end{align}
and $\mathcal{I}_{k}\coloneqq\mathcal{I}_{k}^{+}\cup\mathcal{I}_{k}^{-}$.
This has the effect of excluding a strip around the equator of the Fermi ball, where the number of momenta
per patch may become too small. Note that $\delta>0$ coincides
with the parameter in step 1.
\item Define the collective almost-bosonic creation operator and its normalization
factor as
\begin{equation}
c_{\alpha}^{\ast}(k)\coloneqq
\begin{cases}
  b_{\alpha}^{\ast}(+ k) & \text{if } \alpha\in\mathcal{I}_{k}^{+},\\
  b_{\alpha}^{\ast}(- k) & \text{if } \alpha\in\mathcal{I}_{k}^{-}.
  \end{cases},\ \ 
n_{\alpha}(k)\coloneqq 
\begin{cases}
m_{\alpha}(k) & \text{if } \alpha\in\mathcal{I}_{k}^{+},\\
m_{\alpha}(- k) & \text{if } \alpha\in\mathcal{I}_{k}^{-}.
\end{cases}
\end{equation}
with 
\begin{equation}
b_{\alpha}^{\ast}(k)\coloneqq\frac{1}{m_{\alpha}(k)}\sum_{\substack{p\in B_{F}^{c}\cap B_{\alpha},\\p-k\in B_{F}\cap B_{\alpha}}}a_{p}^{\ast}a_{p-k}^{\ast}
, \quad m_{\alpha}(k)^{2}\coloneqq \sum_{\substack{p\in B_{F}^{c}\cap B_{\alpha},\\p-k\in B_{F}\cap B_{\alpha}}} 1 \end{equation}
being sensitive to being on the north or south hemisphere. The creation
operator can be seen as collective in the sense that it involves a
superposition of all possible fermion pairs with relative momentum
$k$.
\end{enumerate}
Similarly to \eqref{eq:short-hand1} introduced in the previous subsection,
we define as in \cite{Benedikter2021}
\begin{equation}
c^{\ast}(\eta)\coloneqq\sum_{\kinG}\sumak\overline{\eta_{\alpha}(k)}\ccak\label{eq:short-hand2}
\end{equation}
with inner product $\langle\eta,\varphi\rangle\coloneqq\sum_{\kinG}\sumak\overline{\eta_{\alpha}(k)}\varphi_{\alpha}(k)$
for all $\eta,\varphi\in\bigoplus_{\kinG}l^{2}(\mathcal{I}_{k})$.

The following statements hold as a consequence of the above construction.
\begin{itemize}
\item The surface area of a patch satisfies $\sigma(p_{\alpha})\in\mathcal{O}(1/M)$.
\item The Canonical Commutation Relations (CCR) are satisfied up to
an error term (see \cite[Lemma 4.1]{Benedikter2019}): It holds for
all $k',\kinG$ and $\alpha\in\mathcal{I}_{k},\beta\in\mathcal{I}_{k'}$
\begin{align}
[c_{\alpha}(k),c_{\beta}(k')] & =[c_{\alpha}^{\ast}(k),c_{\beta}^{\ast}(k')]=0,\\{}
[c_{\alpha}(k),c_{\beta}^{\ast}(k')] & =\delta_{\alpha,\beta}\big(\delta_{k,k'}+\mathcal{E}_{\alpha}(k,k')\big)\label{eq:approx-CAR}
\end{align}
satisfying $\mathcal{E}_{\alpha}(k,k)\leq0$, $\mathcal{E}_{\alpha}(k,l)=\mathcal{E}_{\alpha}(l,k)^{\ast}$
and
\begin{equation}
\forall\psi\in\mathcal{F}:\qquad\|\mathcal{E}_{\alpha}(k,k')\psi\|\leq\frac{2}{n_{\alpha}(k)n_{\alpha}(k')}\|\mathcal{N}\psi\|.
\end{equation}
\item The almost-bosonic operators change the number operator by two (see
\cite[Lemma 2.3]{Benedikter2019}) in the following sense 
\begin{align}
\cak\mathcal{N} & =(\mathcal{N}+2)\cak.\label{eq:number-op-vs-ccak}
\end{align}
\item The normalization constant satisfies (see \cite[Proposition 3.1]{Benedikter2019})
\begin{equation}
n_{\alpha}(k)^{2}=\frac{4\pi\kF^{2}}{M}|k\cdot\hat{\omega}_{\alpha}|\big(1+o(1)\big).
\end{equation}
\end{itemize}
Also note that the summation in the definition of the almost-bosonic
operators $\ccak$ and $\cak$ involves only finite sets. Unlike in
the exactly bosonic case our almost-bosonic operators therefore inherit
boundedness from the fermionic constituents which satisfy $\|a_{k}^{\ast}\|=\|a_{k}\|=1$.
Subtle questions about the mutual adjointness and domain of the almost-bosonic
operators remain trivial in our case.

\section{Main results}
\subsection{Effective time evolution}
We are now focusing on the effective time evolution of the initial
state $\psi_{0}=R\psi\equiv\phi\otimes R\Omega\in\domH$, i.e. a uncorrelated
product state with the non-interacting Fermi gas prepared as its non-degenerate
ground state with no initial excitations. This set-up corresponds
to a system where the impurity does not interact with the cold Fermi
gas at time $t=0$. Over time, we expect that the influence of the
impurity particle creates and annihilates excitations of the Fermi
ball. Therefore we will use the particle-hole transformation as defined
in \eqref{eq:particle-hole-trafo} to connect the microscopic description
to the following effective description: Let
\begin{align}
\mathbb{H}^{\text{eff}}\coloneqq & -\beta\Delta_{y}+\sum_{\kinG}\sumak\epsilon_{\alpha}(k)\ccak\cak+E_{N}^{\text{pw}}\nonumber \\
 & +\lambda\sum_{\kinG}\sumak\hat{V}(k)e^{iky}n_{\alpha}(k)\bigg(\ccak+c_{\alpha}(-k)\bigg)\label{eq:effective-Hamiltonian}
\end{align}

be our effective Hamiltonian with $\epsilon_{\alpha}(k)=2 k_{F}|k\cdot\omega_{\alpha}|$
and $E_{N}^{\text{pw}}$ as defined in \eqref{eq:E-pw}. We introduce
for all $\kinG$ and $\alpha\in\mathcal{I}_{k}$ 
\begin{equation}
(h_{y})_{\alpha}(k)\coloneqq\lambda\hat{V}(k)e^{iky}n_{\alpha}(k).
\end{equation}
Note that the effective Hamiltonian acts on the components of the
Hilbert space $\domH$ in the sense that we can write

\begin{equation}
\mathbb{H}^{\text{eff}}=(-\beta\Delta_{y})\otimes1+1\otimes\mathbb{D}_{\text{B}}+\Phi(h_{y})+E_{N}^{\text{pw}}
\end{equation}
with $\Phi(h_{y})\coloneqq c^{\ast}(h_{y})+c(h_{y})$ and 
\begin{align}
c^{\ast}(h_{y}) & \coloneqq\sum_{\kinG}\sumak(h_{y})_{\alpha}(k)\ \ccak,\\
c(h_{y}) & \coloneqq\sum_{\kinG}\sumak\overline{(h_{y})_{\alpha}(k)}\ \cak,\\
\mathbb{D}_{\text{B}} & \coloneqq\sum_{\kinG}\sumak\epsilon_{\alpha}(k)\ccak\cak\quad\text{with }\epsilon_{\alpha}(k)=2 k_{F}|k\cdot\omega_{\alpha}|
\end{align}
describing the kinetic energy of the almost-bosonic pair excitations
with linear dispersion relation. Note that by the Kato-Rellich theorem
the effective Hamiltonian \eqref{eq:effective-Hamiltonian} is self-adjoint
in its natural domain and generates a unitary time evolution. 

To state an effective description of the time evolution of those excitations
we compare the particle-hole transformed microscopic dynamics $R^{\ast}e^{-i\mathbb{H}t}R\psi$
with the effective time evolution $e^{-i\mathbb{H}^{\textnormal{eff}}t}\psi$
in Hilbert space norm:
\begin{thm}[Effective dynamics of the system]
\label{thm:Main-thm}Assume that $\hat V \geq 0$ is compactly supported and satisfies $\hat V (-k)= \hat V (k)$
for all $k \in \Z $. Let $ \lambda \in [\kF^{-1/6}, 1] $ be the interaction parameter as introduced in \eqref{eq:microscopic-Hamiltonian} and take the number of patches to be $M=N^{\frac{16}{45}}$ with
$\delta=\frac{2}{15}$ as introduced in \eqref{eq:condition-on-M}. Then it holds for the initial state $\psi=\phi\otimes\Omega\in\domH$
that there is a $C >0$ depending only on the interaction $V$ such that for all $\kF\geq 2$ and
$t\geq0$

\[
\|R^{\ast}e^{-i\mathbb{H}t}R\psi-e^{-i\mathbb{H}^{\textnormal{eff}}t}\psi\|\leq C \left(e^{C \lambda\kF t}-1\right)\kF^{-\frac{1}{5}}.
\]
\end{thm}

\begin{rem}
 Note that the right hand side of the bound is indeed small as long as $t \in \mathcal{O} (\kF^{-1}\lambda^{-1})$. The error $\kF^{-1/5}$ 
 is a result of the optimized choice of $M$ and $\delta$. The non-optimal error bound depends
 on the patch parameters and is given by 
 \begin{equation}
  C\left((1+\lambda^{-1})\lambda^{-1} (M^{-\frac{1}{2}} 
  + MN^{-\frac{2}{3} +\delta})  
  + (N^{-\frac{\delta}{2}}
  +N^{-\frac{1}{6}}M^{\frac{1}{4}})\right) \left(e^{C\lambda\kF t}-1\right). \\ 
 \end{equation}
\end{rem}

\subsection{Effective coherent state}

The effective time evolved state from \prettyref{thm:Main-thm} can be further simplified. Our second result shows how the dynamics can be approximated on the
level of states. 
In order to state the second main result, we introduce almost-bosonic coherent states.

Note that since $B\coloneqq c^{\ast}(\eta)-c(\eta)$ defines for all $\eta\in\bigoplus_{\kinG}l^{2}(\mathcal{I}_{k})$
a bounded operator and satisfies $B=-B^{\ast}$, the exponential operator
$e^{B}$ is well-defined and is unitary.
\begin{defn}
\label{def:Weyl-op}Define for $\eta\in\bigoplus_{\kinG}l^{2}(\mathcal{I}_{k})$
the Weyl operator 
\begin{equation}
W(\eta)\coloneqq e^{B}\coloneqq e^{c^{\ast}(\eta)-c(\eta)}.
\end{equation}
If $\eta\equiv\eta^{y}$ is additionally a bounded multiplication
operator for each $y\in\R^{3}$, we call $W(\eta)\phi\otimes\Omega$
a coupled coherent state with $\phi\otimes\Omega\in\domH$.
\end{defn}

\begin{rem}
If $c^{\ast},c$ would satisfy the CCR without error, one could use
the Baker-Campbell-Hausdorff formula to formally write 
\begin{align}
W(\eta)\phi\otimes\Omega= & e^{-\|\eta\|^{2}/2}e^{c^{\ast}(\eta)}\phi\otimes\bigg\{1,0,\ldots,0,\ldots\bigg\}\nonumber \\
= & e^{-\|\eta\|^{2}/2}e^{c^{\ast}(\eta)}\bigg\{\phi,0,\ldots,0,\ldots\bigg\}\nonumber \\
= & e^{-\|\eta_{s}\|^{2}/2}\bigg\{\phi,\eta\phi,\frac{\eta^{\otimes2}\phi}{\sqrt{2!}},\ldots,\frac{\eta^{\otimes n}\phi}{\sqrt{n!}},\ldots\bigg\}
\end{align}
i.e. the coupled coherent state corresponds to a superposition of
different particle number.  
\end{rem}

\begin{rem}
The coupled coherent state from \prettyref{def:Weyl-op} satisfies the following well-known properties of the Weyl operator
(cf. for example \cite[Chapter 3]{Benedikter2016} or \cite[Appendix A]{Frank2017}) up to certain error terms:
  \begin{itemize}
    \item (shift property) For $\eta,\xi \in\bigoplus_{\kinG}l^{2}(\mathcal{I}_{k})$ it holds
      \begin{equation}
        c(\xi)W(\eta) \phi\otimes\Omega \simeq \sum_{\kinG} \sumak \overline{\xi_{\alpha}(k)} \eta_{\alpha}(k) \phi\otimes\Omega, \label{eq:heuristic-shift-prop}
      \end{equation}

      \item (expectation of the number operator) For $\eta \in\bigoplus_{\kinG}l^{2}(\mathcal{I}_{k})$ it holds
      \begin{equation}
        \langle W(\eta)\phi\otimes\Omega, \mathcal{N}W(\eta)\phi\otimes\Omega \rangle \simeq 2 \|\eta\|^2, \label{eq:heuristic-number-exp}
      \end{equation}

      \item (time derivative of the Weyl operator) For $\eta_{t}\in\bigoplus_{\kinG}l^{2}(\mathcal{I}_{k})$
       differentiable in $t$ with derivative $\dot{\eta_{t}}\in\bigoplus_{\kinG}l^{2}(\mathcal{I}_{k})$ it holds:
      for all $t\in\R$:
      \begin{equation}
        \partial_{t}W(\eta_{t}) \simeq \left(c^{\ast}(\dot{\eta}_{t})-c(\dot{\eta}_{t})+i\textnormal{Im}\langle\dot{\eta}_{t},\eta_{t}\rangle\right)W(\eta_{t}), \label{eq:heuristic-Weyl-deriv}
      \end{equation}
  \end{itemize}
We give rigorous statements on the error terms and proofs of the approximate properties in \prettyref{lem:Approximate-shift}, \prettyref{prop:expectation-number-W},
\prettyref{lem:derivative-of-Weyl-op} of \prettyref{sec:Coheren-Proof}.
\end{rem}

Consider now the following state for all times $t\in\R$
\begin{align}
\psi_{t} & \coloneqq e^{iP(t)}W(\eta_{t})\phi\otimes\Omega\\
\text{with }P(t) & =2\text{Im}(\nu_{t})-E_{N}^{\text{pw}}t-\text{Im}\int_{0}^{t}\d s\langle\dot{\eta}_{s},\eta_{s}\rangleka \label{eq:coherent-phase}
\end{align}
with the choices of 
\begin{align}
\big(\eta_{s}\big)_{\alpha}(k)\coloneqq & \big(\eta_{s}^{y}\big)_{\alpha}(k)\coloneqq\frac{e^{-is\epsilon_{\alpha}(k)}-1}{\epsilon_{\alpha}(k)}(h_{y})_{\alpha}(k)=\frac{e^{-is\epsilon_{\alpha}(k)}-1}{\epsilon_{\alpha}(k)}\lambda\hat{V}(k)n_{\alpha}(k)e^{iky},\label{eq:Coherent-state-variableI}\\
\nu_{s} \coloneqq & \sum_{\kinG} \sumak \frac{e^{-is\epsilon_{\alpha}(k)}+is\epsilon_{\alpha}(k)-1}{\epsilon_{\alpha}(k)^{2}}|(h_{y})_{\alpha}(k)|^{2}\label{eq:Coherent-state-variableII}
\end{align}
for all $\kinG,\alpha\in\mathcal{I}_{k}$. Due to $\big(\eta_{s}\big)_{\alpha}(k)=-ie^{-is\epsilon_{\alpha}(k)/2}\frac{\text{sin}\left(\epsilon_{\alpha}(k)s/2\right)}{\epsilon_{\alpha}(k)/2}(h_{y})_{\alpha}(k)$
and \prettyref{lem:Approx-n_alpha(k)^2} the norm is bounded for all
$s\in\R$
\begin{align}
\|\eta_{s}\| & \equiv\bigg(\sum_{\kinG}\sumak|\big(\eta_{s}\big)_{\alpha}(k)|^{2}\bigg)^{1/2}\leq\min\left\{ \sqrt{\pi}\|\hat{V}(\cdot)^{1/2}\|_{2} \lambda \kF s,\sqrt{2\pi}\|\hat{V}\|_{2} \lambda \log(4\kF s+2)\right\} \label{eq:eta-bound-simple} 
\end{align}
as shown later in \prettyref{lem:eta-bounds} and similarly for $| \nu_s |$.

In particular, it holds $(\eta_{0},\nu_{0})=(0,0)$ and $\lim_{\epsilon_{\alpha}(k)\to0}(\eta_{s},\nu_{s})=(-is h_y,0)$
and therefore $\psi_{0}=\phi\otimes\Omega$. Thus, as mentioned before,
$\eta\equiv\eta^{y}$ as defined in \eqref{eq:Coherent-state-variableI}
corresponds to an interaction term with a bounded multiplication operator
acting on $L^{2}(\Lambda,\d y)$. The state $\psi_{t}=e^{iP(t)}W(\eta_{t})\phi\otimes\Omega$
can therefore be seen as a time-dependent coupled coherent state. 

The following theorem states that $\psi_{t}$ approximately corresponds
to the effective time evolution generated by $\Heff$.
\begin{thm}[Effective coherent dynamics] \label{thm:Effective-coherent-dynamics}
Consider the initial state 
$\psi=\phi\otimes\Omega\in\domH$ with $\sum_{i=1}^{3}\big(\|\partial_{y_{i}}\phi\|+\|\partial_{y_{i}}^{2}\phi\|\big)\leq c<\infty$
for a constant $c>0$. Under the assumptions of \prettyref{thm:Main-thm}, there exists a constant $C>0$
and a function $Q:\R_{\geq0}\to\R_{\ge0}$ monotonically increasing
with $Q(0)=0$ such that for all $t\geq0$ 
\[
\|e^{-i\mathbb{H}^{\text{eff}}t}\psi-e^{iP(t)}W(\eta_{t})\psi\|\leq C Q(\lambda \kF t)\max\{\kF^{-\frac{2}{15}},c\beta t\}
\]
with $P(t)$ given by \eqref{eq:coherent-phase}, $\eta_{t}$ given
by \eqref{eq:coherent-state-error-I}, $C$ and $Q$ depending only on the interaction $V$.
\end{thm}

\begin{rem} \label{rem:coherent-bound}
Note that the upper bound above is indeed meaningful in the sense that
the bound is small for $t \in \mathcal{O} (\kF^{-1}\lambda^{-1})$ and 
$\beta \in o(\lambda \kF) $. The latter condition together with the upper bound for $\|\Delta_{y}W(\eta_{t})\phi\otimes\Omega\|$
from \prettyref{lem:Laplacian-estimate} ensures that the contribution from the kinetic energy term 
$h_0=-\beta \Delta_y$ of the impurity particle remains negligible on the relevant time scale. This seems to be
crucial in our approach since otherwise $h_0$ would generate non-trivial correlations
and make the coherent state form inapplicable. 
\end{rem}

\begin{rem}
Due to the explicit formulation of the coupled coherent
state provided in \eqref{eq:Coherent-state-variableI}, we can quantify
the number of collective excitations over time. More concretely, using
\eqref{eq:heuristic-number-exp} the term $\|\eta_{t}\|^2$, which
is calculated in \prettyref{lem:eta-bounds}, represents the expected
number of excitations generated by the interaction with the impurity.
Assuming that the impurity, together with its excitations, can
be interpreted as polaron-like quasi-particle, the quantity $\|\eta_{t}\|^2$
offers insight into the quasi-particle formation process. As displayed in \prettyref{fig:eta_t} the graph shows a parabolic growth followed by a logarithmic increase as qualitative feature.

\begin{figure}[h]  
  \centering
  \includegraphics[width=1\linewidth]{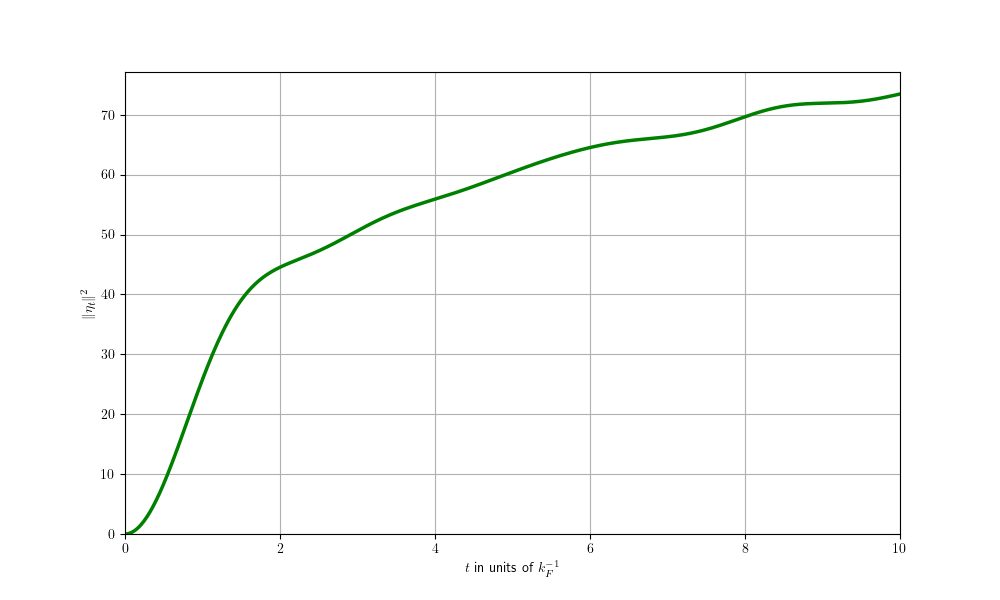}
  \caption{Plot of $\| \eta_t \|^2 $ using \prettyref{lem:eta-bounds} with constant $\hat V$  and $\lambda=1$. 
  As a qualitative feature one observes a parabolic growth followed by a logarithmic increase.}
  \label{fig:eta_t}
\end{figure}

\end{rem}

\begin{rem}
\label{rem:Heuristic-calculation}The proof is based on the following
observation: It holds by virtue of Duhamel's formula
\begin{align}
 & \|e^{-i\mathbb{H}^{\text{eff}}t}\psi-e^{iP(t)}W(\eta_{t})\psi\|\nonumber \\
 & =\|\psi-e^{i\mathbb{H}^{\text{eff}}t}e^{-iE_{N}^{\text{pw}}t}e^{i2\text{Im}(\nu_{t})}e^{-i\text{Im}\int_{0}^{t}\d s\langle\dot{\eta}_{s},\eta_{s}\rangle}W(\eta_{t})\psi\|\\
 & =\|\int_{0}^{t}\d s\ e^{i\mathbb{H}^{\text{eff}}s}e^{iP(s)}\left\{ \big(\mathbb{H}^{\text{eff}}-E_{N}^{\text{pw}}+2\text{Im}(\dot{\nu}_{s})-\text{Im}\langle\dot{\eta}_{s},\eta_{s}\rangle\big)W(\eta_{s})\psi-i\partial_{s}W(\eta_{s})\psi\right\} \|\\
 & \leq\int_{0}^{t}\d s\|\big(\mathbb{H}^{\text{eff}}-E_{N}^{\text{pw}}+2\text{Im}(\dot{\nu}_{s})-\text{Im}\langle\dot{\eta}_{s},\eta_{s}\rangle\big)W(\eta_{s})\psi-i\partial_{s}W(\eta_{s})\psi\|.
\end{align}

If the collective operators $\ccak$ and $\cak$ were exactly
bosonic, i.e. the CCR held without error, we could use the shift property
\eqref{eq:heuristic-shift-prop} of the Weyl operator to commute
$\cak$ to the vacuum $\Omega$ with the cost of some inner product
terms. In this case we would observe with the short-hand notation $c^{\ast}c(\epsilon)\coloneqq\sum_{\kinG}\sumak\epsilon_{\alpha}(k)\ccak\cak$
and by applying \eqref{eq:heuristic-Weyl-deriv} that 
\begin{align}
 & \big[\mathbb{H}^{\text{eff}}-E_{N}^{\text{pw}}+2\text{Im}(\dot{\nu}_{s})-i\left\{ c^{\ast}(\dot{\eta}_{s})-c(\dot{\eta}_{s})\right\} \big]W(\eta_{s})\psi\nonumber \\
= & \big[h_{0}+c^{\ast}c(\epsilon)+c^{\ast}(h_{y})+c(h_{y})+2\text{Im}(\dot{\nu}_{s})+c^{\ast}(-i\dot{\eta}_{s})-c(i\dot{\eta}_{s})\big]W(\eta_{s})\phi\otimes\Omega\nonumber \\
= & \big[h_{0}+c^{\ast}(\epsilon\eta_{s})+c^{\ast}(h_{y})+\langle h_{y},\eta_{s}\rangle+2\text{Im}(\dot{\nu}_{s})+c^{\ast}(-i\dot{\eta}_{s})-\langle i\dot{\eta}_{s},\eta_{s}\rangle\big]W(\eta_{s})\phi\otimes\Omega\nonumber \\
= & h_{0}W(\eta_{s})\phi\otimes\Omega\label{eq:exactly-bosonic-equality}
\end{align}
is exactly vanishing up to the kinetic energy term $h_{0}$ of the
impurity particle since $\eta_{t}$ as defined in \eqref{eq:Coherent-state-variableI}
solves the ODE $\epsilon\eta_{t}+h_{y}=i\dot{\eta}_{t}$ and $\nu_{t}$
as defined in \eqref{eq:Coherent-state-variableII} absorbs all scalar
terms. 
\end{rem}

\begin{rem}\label{rem:Fermi-blockade}
Note that as long as $\|\eta_t\|$ is of order 1, the coupled coherent state is different from a 
free decoupled dynamics with $\psi_t^{\textnormal{free}} \sim e^{i\beta \Delta_yt} \phi \otimes \Omega$.
Vice versa, if $\|\eta_t\| \in o(1)$ in terms of $\kF$ one can easily see with a Duhamel argument, analogue to the previous remark,
that $\|\psi_t^{\textnormal{free}}-e^{iP(t)}W(\eta_t)\psi\| \to 0$ for large $\kF$. This is because all the terms appearing in 
$\partial_t W(\eta_t)$ from \eqref{eq:heuristic-Weyl-deriv} can be bounded in terms of $\|\eta_t\|$.
Using \prettyref{eq:eta-bound-simple} we can identify the following cases:
\begin{itemize}
\item $\lambda = 1, t \in o(\kF^{-1})$: $\|\eta_t\| \in o(1)$ with respect to $\kF$ \\ 
The time scale is too short for forming excitations thus $\psi_t^{\textnormal{free}}$ 
is a good approximation. Note that this is compatible with \cite{mitrouskas2021effective} as mentioned in \prettyref{ssec:couplings}.
\item $\lambda =1, t \in \mathcal{O}(\kF^{-1}) $: $\|\eta_t\| \in \mathcal{O}(1)$ with respect to $\kF$ \\ 
On this time scale exitations of the Fermi gas become important and thus the free decoupled evolution is not a good approximation anymore.
We refer to the proof of \prettyref{cor:Corollary} on
how to show rigorously that the derived effective description is relevant on this time scale.
\item $\lambda \in o(1), 0 \leq t \in \mathcal{O} (\kF^{-1}\lambda^{-1})$: $\|\eta_t\| \in o(1)$ with respect to $\kF$ \\ 
Also here, $\psi_t^{\textnormal{free}}$ is a good description. Note that the number of excitations grows only logarithmically in time.
Thus for any $\lambda \in o(1)$ on gets $\|\eta_t\| \in o(1)$. In the case of $\lambda \in o(1)$
we expect this result of free decoupling to hold for even longer times, possibly of order 1.
In order to show this expected behavior, we expect a method different from ours to be more suitable, for example
a perturbative expansion in the spirit of \cite{mitrouskas2021effective}
to higher orders to exploit the small coupling.
\end{itemize}
\end{rem}

Now by virtue of the previous statement we are able to show that the
linear coupling term of the effective Hamiltonian is essential for
the effective description and cannot be neglected on an approximate
level for $\lambda=1$ and $t\in \mathcal O(\kF^{-1})$. 

Let 
\begin{equation}
\widetilde{\mathbb{H}^{\text{eff}}}\coloneqq\Heff-c^{\ast}(h_{y})-c(h_{y})\equiv\beta(-\Delta_{y})+\sum_{\kinG}\sumak\epsilon_{\alpha}(k)\ccak\cak+E_{N}^{\text{pw}}
\end{equation}
be the effective Hamiltonian without the linear coupling.
\begin{cor}\label{cor:Corollary}
Under the assumptions of \prettyref{thm:Effective-coherent-dynamics} with $\lambda =1$ and a $T>0$, there
exists a monotonically increasing function $C:\R_{\geq0}\to\R_{\geq0}$ only depending on $V$
such that for all $t\in[0,T]$
\[
\|R^{\ast}e^{-i\mathbb{H}t}R\psi-e^{-i\widetilde{\mathbb{H}^{\textnormal{eff}}}t}\psi\|\geq C(t)
-\mathcal{O}\left(\max\{\kF^{-\frac{2}{15}},\beta\kF^{-1}\}\right).
\]
In particular it holds $C(t)\in\mathcal{O}(1)$ with respect to $\kF$ for
all $t \in \mathcal O(\kF^{-1})$.
\end{cor}

\section{Proof of \prettyref{thm:Main-thm}}

In all of our estimates, we need to control the number operator $\mathcal{N}$
acting on the time evolved state $\psi_{t}=e^{-i\mathbb{H}^{\textnormal{eff}}t}\phi\otimes\Omega$.
The following statement shows that the effective time evolution preserves
the order of magnitude of the number of excitations: 
\begin{prop}[Effective time evolution of the number operator]
\label{prop:Effective-time-evolution-number-op}There exists a constant
$C>0$ only depending on $V$ such that it holds for all $t\in\R,n\in\mathbb{N}$
and $\psi\in\domH$
\[
\langle e^{-i\mathbb{H}^{\textnormal{eff}}t}\psi,(\mathcal{N}+1)^{n}e^{-i\mathbb{H}^{\textnormal{eff}}t}\psi\rangle\leq e^{nC\lambda\kF t}\langle\psi,(\mathcal{N}+3)^{n}\psi\rangle.
\]
\end{prop}


\begin{proof}
We want to use Grönwall's lemma and therefore estimate the derivative
\begin{align}
 & \bigg|i\partial_{t}\langle e^{-i\mathbb{H}^{\text{eff}}t}\psi,(\mathcal{N}+3)^{n}e^{-i\mathbb{H}^{\text{eff}}t}\psi\rangle\bigg|\nonumber \\
\leq & \bigg|\langle e^{-i\mathbb{H}^{\text{eff}}t}\psi,\sum_{j=0}^{n-1}(\mathcal{N}+3)^{j}[\mathcal{N},\mathbb{H}^{\text{eff}}](\mathcal{N}+3)^{n-j-1}e^{-i\mathbb{H}^{\text{eff}}t}\psi\rangle\bigg|\nonumber \\
= & \bigg|2\sum_{\kinG}\sumak\phiak\sum_{j=0}^{n-1}\langle e^{-i\mathbb{H}^{\text{eff}}t}\psi,(\mathcal{N}+3)^{j}\big(\ccak-\cak\big)(\mathcal{N}+3)^{n-j-1}e^{-i\mathbb{H}^{\text{eff}}t}\psi\rangle\bigg|
\end{align}

We split the difference of $\big(\ccak-\cak\big)$ and consider the
term with $\ccak$ first. Insert here $\id=(\mathcal{N}+1)^{\frac{n}{2}-1-j}(\mathcal{N}+1)^{j+1-\frac{n}{2}}$
between $(\mathcal{N}+3)^{j}$ and $\ccak$ and use the commutation
$\mathcal{N}\ccak=\ccak(\mathcal{N}+2)$ to obtain 
\begin{equation}
(\mathcal{N}+3)^{j}\ccak(\mathcal{N}+3)^{n-j-1}=(\mathcal{N}+3)^{j}(\mathcal{N}+1)^{\frac{n}{2}-1-j}\ccak(\mathcal{N}+3)^{\frac{n}{2}}.
\end{equation}

We introduce the notation $\xi_{j}\coloneqq(\mathcal{N}+1)^{\frac{n}{2}-1-j}(\mathcal{N}+3)^{j}e^{-i\mathbb{H}^{\text{eff}}t}\psi$
and $\tilde{\xi}\coloneqq(\mathcal{N}+3)^{\frac{n}{2}}e^{-i\mathbb{H}^{\text{eff}}t}\psi$
 to estimate
\begin{align}
 & \bigg|\sum_{\kinG}\sumak\phiak\sum_{j=0}^{n-1}\langle\xi_{j},\ccak\tilde{\xi}\rangle\bigg|\nonumber \\
\leq & \sum_{\kinG}\sumak|\phiak|\sum_{j=0}^{n-1}\|\cak\xi_{j}\|\ \|\tilde{\xi}\|\nonumber \\
\leq & C\lambda\sum_{\kinG}|\hat{V}(k)|\bigg(\sumak n_{\alpha}(k)^{2}\bigg)^{1/2}\sum_{j=0}^{n-1}\bigg(\sumak\|\cak\xi_{j}\|^{2}\bigg)^{1/2}\|\tilde{\xi}\|\nonumber \\
\leq & C\lambda\kF\sum_{\kinG}|\hat{V}(k)|\sum_{j=0}^{n-1}\|\mathcal{N}^{1/2}\xi_{j}\|\tilde{\xi}\|\nonumber \\
\leq & C\lambda\kF\|\hat{V}\|_{1}\sum_{j=0}^{n-1}\|(\mathcal{N}+1)\xi_{j}\|\ \|\tilde{\xi}\|\nonumber \\
\leq & Cn\lambda\kF\|\hat{V}\|_{1}\langle e^{-i\mathbb{H}^{\text{eff}}t}\psi,(\mathcal{N}+3)^{n}e^{-i\mathbb{H}^{\text{eff}}t}\psi\rangle
\end{align}
where we used \prettyref{lem:Pair-operator-bounds} and \prettyref{lem:Approx-n_alpha(k)^2-wo-sum}
in the fourth line and the operator inequality $\mathcal{N}^{1/2}\leq(\mathcal{N}+1)$
in the fifth line. Note that $|k|<C$ for all $\kinG$ since $\Gamma \subset \supp \hat{V}$ and $\hat V$ has bounded support by assumption.

The second term with $\cak$ can be treated by inserting $\id=(\mathcal{N}+1)^{\frac{n}{2}-j}(\mathcal{N}+1)^{j-\frac{n}{2}}$
and using the commutation $\cak\mathcal{N}=(\mathcal{N}+2)\cak$.
\begin{equation}
(\mathcal{N}+3)^{j}\cak(\mathcal{N}+3)^{n-j-1}=(\mathcal{N}+3)^{\frac{n}{2}}\cak(\mathcal{N}+1)^{j-\frac{n}{2}}(\mathcal{N}+3)^{n-j-1}.
\end{equation}

We introduce the notation $\chi_{j}\coloneqq(\mathcal{N}+1)^{j-\frac{n}{2}}(\mathcal{N}+3)^{n-j-1}e^{-i\mathbb{H}^{\text{eff}}t}\psi$
\begin{align}
 & \bigg|\sum_{\kinG}\sumak\phiak\sum_{j=0}^{n-1}\langle\tilde{\xi},\cak\chi_{j}\rangle\bigg|\nonumber \\
\leq & \sum_{\kinG}\sumak|\phiak|\sum_{j=0}^{n-1}\|\cak\chi_{j}\|\ \|\tilde{\xi}\|\nonumber \\
\leq & C\lambda\sum_{\kinG}|\hat{V}(k)|\bigg(\sumak n_{\alpha}(k)^{2}\bigg)^{1/2}\sum_{j=0}^{n-1}\bigg(\sum_{\alpha}\|\cak\chi_{j}\|^{2}\bigg)^{1/2}\|\tilde{\xi}\|\nonumber \\
\leq & C\lambda\kF\sum_{\kinG}|\hat{V}(k)|\sum_{j=0}^{n-1}\|\mathcal{N}^{1/2}\chi_{j}\|\tilde{\xi}\|\nonumber \\
\leq & C\lambda\kF\|\hat{V}\|_{1}\sum_{j=0}^{n-1}\|(\mathcal{N}+1)\chi_{j}\|\ \|\tilde{\xi}\|\nonumber \\
\leq & Cn\lambda\kF\|\hat{V}\|_{1}\langle e^{-i\mathbb{H}^{\text{eff}}t}\psi,(\mathcal{N}+3)^{n}e^{-i\mathbb{H}^{\text{eff}}t}\psi\rangle.
\end{align}
Altogether it follows with the Grönwall's lemma that 
\begin{equation}
\langle e^{-i\mathbb{H}^{\text{eff}}t}\psi,(\mathcal{N}+3)^{n}e^{-i\mathbb{H}^{\text{eff}}t}\psi\rangle\leq\exp\big(Cn\lambda\kF\|\hat{V}\|_{1}t\big)\langle\psi,(\mathcal{N}+3)^{n}\psi\rangle
\end{equation}
which is the desired result since $\|\hat{V}\|_{1}<C$.
\end{proof}
The following statement shows that the fermionic kinetic energy term
\prettyref{eq:H_0-def} can be approximated by the almost-bosonic
kinetic energy term.
\begin{prop}[Approximation of the kinetic energy]
\label{prop:Approx-kinetic-energy}There exists a constant $C>0$ only depending on $V$
such that it holds for all $t\in\mathbb{R}$ and $\psi\in\domH$  
\begin{align}
 & \left|\|(\mathbb{H}_{0}-\mathbb{D}_{B})e^{-i\mathbb{H}^{\textnormal{eff}}t}\psi\|-\|(\mathbb{H}_{0}-\mathbb{D}_{B})\psi\|\right|\nonumber \\
 & \leq C(1+\lambda^{-1}) N^{\frac{1}{3}}\left(e^{C\lambda\kF t}-1\right)\left(M^{-\frac{1}{2}}\|(\mathcal{N}+3)\psi\|+\kF MN^{-1+\delta}\|(\mathcal{N}+3)^{2}\psi\|\right)\label{eq:Approx-kinetic-energy}
\end{align}
using \prettyref{prop:Effective-time-evolution-number-op}.
\end{prop}

\begin{rem}
Note that in the case of $\psi\equiv\phi\otimes\Omega$ it holds $(\mathbb{H}_{0}-\mathbb{D}_{B})\psi=0$
and $\|(\mathcal{N}+3)^{2}\psi\|=9\leq C$. Also note that $\left(e^{C\lambda\kF t}-1\right)=\lambda\kF t+\mathcal{O}((\lambda\kF t)^{2})$. 
\end{rem}

\begin{proof}
It holds 
\begin{align}
 & \left|\partial_{t}\|(\mathbb{H}_{0}-\mathbb{D}_{B})e^{-i\mathbb{H}^{\text{eff}}t}\psi\|^{2}\right|\nonumber \\
 & =\left|\partial_{t}\langle e^{-i\mathbb{H}^{\text{eff}}t}\psi,(\mathbb{H}_{0}-\mathbb{D}_{B})^{2}e^{-i\mathbb{H}^{\text{eff}}t}\psi\rangle\right|\nonumber \\
 & =\left|\langle e^{-i\mathbb{H}^{\text{eff}}t}\psi,[(\mathbb{H}_{0}-\mathbb{D}_{B})^{2},\mathbb{H}^{\text{eff}}]e^{-i\mathbb{H}^{\text{eff}}t}\psi\rangle\right|\nonumber \\
 & =\left|\langle e^{-i\mathbb{H}^{\text{eff}}t}\psi,\left((\mathbb{H}_{0}-\mathbb{D}_{B})[(\mathbb{H}_{0}-\mathbb{D}_{B}),\mathbb{H}^{\text{eff}}]+[(\mathbb{H}_{0}-\mathbb{D}_{B}),\mathbb{H}^{\text{eff}}](\mathbb{H}_{0}-\mathbb{D}_{B})\right)e^{-i\mathbb{H}^{\text{eff}}t}\psi\rangle\right|\nonumber \\
 & \leq2\left|\langle e^{-i\mathbb{H}^{\text{eff}}t}\psi,(\mathbb{H}_{0}-\mathbb{D}_{B})[(\mathbb{H}_{0}-\mathbb{D}_{B}),\mathbb{H}^{\text{eff}}]e^{-i\mathbb{H}^{\text{eff}}t}\psi\rangle\right| \\
\end{align}

and make use of 
\begin{align}
[(\mathbb{H}_{0}-\mathbb{D}_{B}),\ccak] & \eqqcolon\mathfrak{E}_{\alpha}^{\text{lin}}(k)^{\ast}-\Ebos_{\alpha}(k)^{\ast}\eqqcolon\mathfrak{E}_{\alpha}(k)^{\ast},\\{}
[(\mathbb{H}_{0}-\mathbb{D}_{B}),\cak] & =-\mathfrak{E}_{\alpha}(k)
\end{align}
to arrive at 
\begin{align}
[(\mathbb{H}_{0}-\mathbb{D}_{B}),\mathbb{H}^{\text{eff}}] & =[(\mathbb{H}_{0}-\mathbb{D}_{B}),c^{\ast}c(\epsilon)+c^{\ast}(h_{y})+c(h_{y})]\\
 & =\langle\mathfrak{E},c(\epsilon)\rangleka-\langle c(\epsilon),\mathfrak{E}\rangleka+\langle\mathfrak{E},h_{y}\rangleka-\langle h_{y},\mathfrak{E}\rangleka
\end{align}
with the short notation $\langle A,B\rangleka\coloneqq\sum_{\kinG}\sumak A_{\alpha}^{\ast}(k)B_{\alpha}(k)$.

Bounds for the error terms are readily provided in \prettyref{lem:Elin-bound}
and \prettyref{lem:Ebos-bound} of the appendix: 
\begin{align}
\sumak\|\Elin_{\alpha}(k)\psi\|^{2} & \leq C\left( N^{\frac{1}{3}}M^{-\frac{1}{2}}\right)^{2}\|(\mathcal{N}+1)^{\frac{1}{2}}\psi\|^{2},\\
\sumak\|\mathfrak{E}_{\alpha}^{\text{B}}(k)\psi\|^{2} & \leq C\left(\kF MN^{-\frac{2}{3}+\delta}\right)^{2}\|(\mathcal{N}+1)^{\frac{3}{2}}\psi\|^{2}.
\end{align}

Furthermore use the bounds $\sumak\ccak\cak\leq\mathcal{N}$, $\epsilon_{\alpha}(k)\leq C\kF$
and that $\Ebos_{\alpha}(k),\Elin_{\alpha}(k),\cak$ all annihilate
exactly two fermions to estimate 
\begin{align}
 & \langle e^{-i\mathbb{H}^{\text{eff}}t}\psi,(\mathbb{H}_{0}-\mathbb{D}_{B})\langle c(\epsilon),\Elin-\Ebos\rangleka e^{-i\mathbb{H}^{\text{eff}}t}\psi\rangle\nonumber \\
\leq & \sum_{\kinG}\sumak\left|\langle\epsilon_{\alpha}(k)\cak(\mathcal{N}+1)^{-1/2}(\mathbb{H}_{0}-\mathbb{D}_{B})e^{-i\mathbb{H}^{\text{eff}}t}\psi,\Elin_{\alpha}(k)(\mathcal{N}+1)^{1/2}e^{-i\mathbb{H}^{\text{eff}}t}\psi\rangle\right|\nonumber \\
 & +\sum_{\kinG}\sumak\left|\langle\epsilon_{\alpha}(k)\cak(\mathcal{N}+1)^{-1/2}(\mathbb{H}_{0}-\mathbb{D}_{B})e^{-i\mathbb{H}^{\text{eff}}t}\psi,\Ebos_{\alpha}(k)(\mathcal{N}+1)^{1/2}e^{-i\mathbb{H}^{\text{eff}}t}\psi\rangle\right|\\
\leq & \sum_{\kinG}\bigg(\sumak\|\epsilon_{\alpha}(k)\cak(\mathcal{N}+1)^{-1/2}(\mathbb{H}_{0}-\mathbb{D}_{B})e^{-i\mathbb{H}^{\text{eff}}t}\psi\|^{2}\bigg)^{1/2}\nonumber \\
 & \quad\times\left( \bigg(\sumak\|\Elin_{\alpha}(k)(\mathcal{N}+1)^{1/2}e^{-i\mathbb{H}^{\text{eff}}t}\psi\|^{2}\bigg)^{1/2}+\bigg(\sumak\|\Ebos_{\alpha}(k)(\mathcal{N}+1)^{1/2}e^{-i\mathbb{H}^{\text{eff}}t}\psi\|^{2}\bigg)^{1/2}\right) \\
\leq & C\kF\|(\mathbb{H}_{0}-\mathbb{D}_{B})e^{-i\mathbb{H}^{\text{eff}}t}\psi\|\left(  N^{\frac{1}{3}}M^{-\frac{1}{2}}\|(\mathcal{N}+1)e^{-i\mathbb{H}^{\text{eff}}t}\psi\|+\kF MN^{-\frac{2}{3}+\delta}\|(\mathcal{N}+1)^{2}e^{-i\mathbb{H}^{\text{eff}}t}\psi\|\right) \\
\leq & C^{2}\kF N^{\frac{1}{3}}\|(\mathbb{H}_{0}-\mathbb{D}_{B})e^{-i\mathbb{H}^{\text{eff}}t}\psi\|\left(M^{-\frac{1}{2}}\|(\mathcal{N}+1)e^{-i\mathbb{H}^{\text{eff}}t}\psi\|+\kF MN^{-1+\delta}\|(\mathcal{N}+1)^{2}e^{-i\mathbb{H}^{\text{eff}}t}\psi\|\right).
\end{align}

The term $\langle c(\epsilon),\mathfrak{E}\rangleka$
can be treated analogously. The other terms can be estimated similarly using $|\phiak|\leq C\lambda|\hat{V}(k)||n_{\alpha}(k)|$
and Cauchy-Schwarz inequality with \eqref{lem:Approx-n_alpha(k)^2-wo-sum}
\begin{align}
 & \langle e^{-i\mathbb{H}^{\text{eff}}t}\psi,(\mathbb{H}_{0}-\mathbb{D}_{B})\langle\Elin-\Ebos,h_{y}\rangleka e^{-i\mathbb{H}^{\text{eff}}t}\psi\rangle\nonumber \\
\leq & C\lambda\sum_{\kinG}|\hat{V}(k)|\sumak\left|\langle n_{\alpha}(k)(\mathbb{H}_{0}-\mathbb{D}_{B})e^{-i\mathbb{H}^{\text{eff}}t}\psi,\Elin_{\alpha}(k)e^{-i\mathbb{H}^{\text{eff}}t}\psi\rangle\right|\nonumber \\
 & +C\lambda\sum_{\kinG}|\hat{V}(k)|\sumak\left|\langle n_{\alpha}(k)(\mathbb{H}_{0}-\mathbb{D}_{B})e^{-i\mathbb{H}^{\text{eff}}t}\psi,\Ebos_{\alpha}(k)e^{-i\mathbb{H}^{\text{eff}}t}\psi\rangle\right|\\
\leq & C\lambda\kF\|(\mathbb{H}_{0}-\mathbb{D}_{B})e^{-i\mathbb{H}^{\text{eff}}t}\psi\|\sum_{\kinG}|\hat{V}(k)|\Bigg\{\bigg(\sumak\|\Elin_{\alpha}(k)e^{-i\mathbb{H}^{\text{eff}}t}\psi\|^{2}\bigg)^{1/2}\\
 & \hphantom{C\lambda\kF\|(\mathbb{H}_{0}-\mathbb{D}_{B})e^{-i\mathbb{H}^{\text{eff}}t}\psi\|\sum_{\kinG}|\hat{V}(k)|\bigg\{}+\bigg(\sumak\|\Ebos_{\alpha}(k)e^{-i\mathbb{H}^{\text{eff}}t}\psi\|^{2}\bigg)^{1/2}\Bigg\}\\
\leq & C\lambda\|\hat{V}\|_{1}\kF\|(\mathbb{H}_{0}-\mathbb{D}_{B})e^{-i\mathbb{H}^{\text{eff}}t}\psi\|\bigg\{ M^{-\frac{1}{2}}N^{\frac{1}{3}}\|(\mathcal{N}+1)^{\frac{1}{2}}e^{-i\mathbb{H}^{\text{eff}}t}\psi\|\\
 & \hphantom{C\lambda\|\hat{V}\|_{1}\kF\|(\mathbb{H}_{0}-\mathbb{D}_{B})e^{-i\mathbb{H}^{\text{eff}}t}\psi\|\bigg\{}+\kF MN^{-\frac{2}{3}+\delta}\|(\mathcal{N}+1)^{\frac{3}{2}}e^{-i\mathbb{H}^{\text{eff}}t}\psi\|\bigg\}\nonumber \\
\leq & C\lambda\kF M^{-\frac{1}{2}}N^{\frac{1}{3}}\|(\mathbb{H}_{0}-\mathbb{D}_{B})e^{-i\mathbb{H}^{\text{eff}}t}\psi\|\bigg\{\|(\mathcal{N}+1)^{\frac{1}{2}}e^{-i\mathbb{H}^{\text{eff}}t}\psi\|\\
 & \hphantom{C\lambda\kF M^{-\frac{1}{2}} N^{\frac{1}{3}}\|(\mathbb{H}_{0}-\mathbb{D}_{B})e^{-i\mathbb{H}^{\text{eff}}t}\psi\|\bigg\{}+\kF M^{\frac{3}{2}}N^{-1+\delta}\|(\mathcal{N}+1)^{\frac{3}{2}}e^{-i\mathbb{H}^{\text{eff}}t}\psi\|\bigg\}.
\end{align}

Thus we derive 
\begin{align}
 & \partial_{t}\|(\mathbb{H}_{0}-\mathbb{D}_{B})e^{-i\mathbb{H}^{\text{eff}}t}\psi\|\nonumber \\
\leq & C\lambda\kF M^{-\frac{1}{2}}N^{\frac{1}{3}}\left(\|(\mathcal{N}+1)e^{-i\mathbb{H}^{\text{eff}}t}\psi\|+\kF M^{\frac{3}{2}}N^{-1+\delta}\|(\mathcal{N}+1)^{2}e^{-i\mathbb{H}^{\text{eff}}t}\psi\|\right)\nonumber \\
 & +C\lambda\kF N^{\frac{1}{3}}\left(\|(\mathcal{N}+1)^{\frac{1}{2}}e^{-i\mathbb{H}^{\text{eff}}t}\psi\|+\kF M^{\frac{3}{2}}N^{-1+\delta}\|(\mathcal{N}+1)^{\frac{3}{2}}e^{-i\mathbb{H}^{\text{eff}}t}\psi\|\right)\\
\leq & C^{2}\kF N^{\frac{1}{3}}\left(e^{C\lambda\kF t}\|(\mathcal{N}+3)\psi\|+\kF M^{\frac{3}{2}}N^{-1+\delta}e^{2C\lambda\kF t}\|(\mathcal{N}+3)^{2}\psi\|\right)\nonumber \\
 & +C\lambda\kF N^{\frac{1}{3}}\left(e^{\frac{1}{2}C\lambda\kF t}\|(\mathcal{N}+3)^{\frac{1}{2}}\psi\|+\kF M^{\frac{3}{2}}N^{-1+\delta}e^{\frac{3}{2}C\lambda\kF t}\|(\mathcal{N}+1)^{\frac{3}{2}}\psi\|\right)
\end{align}
where we used the Grönwall bound from \prettyref{prop:Effective-time-evolution-number-op} in the last inequality.
Now integrating over $t$ yields
\begin{align}
 & \|(\mathbb{H}_{0}-\mathbb{D}_{B})e^{-i\mathbb{H}^{\text{eff}}t}\psi\|-\|(\mathbb{H}_{0}-\mathbb{D}_{B})\psi\|\nonumber \\
\leq & C^{2}\lambda^{-1}M^{-\frac{1}{2}}N^{\frac{1}{3}}\left(e^{C\lambda\kF t}-1\right)\left(\|(\mathcal{N}+3)\psi\|+\kF M^{\frac{3}{2}}N^{-1+\delta}\|(\mathcal{N}+3)^{2}\psi\|\right)\nonumber \\
 & +C M^{-\frac{1}{2}}N^{\frac{1}{3}}\left(e^{C\lambda\kF t}-1\right)\left(\|(\mathcal{N}+3)^{\frac{1}{2}}\psi\|+\kF M^{\frac{3}{2}}N^{-1+\delta}\|(\mathcal{N}+1)^{\frac{3}{2}}\psi\|\right)
\end{align}
which corresponds to the desired result.
\end{proof}
We are now ready to give the proof of the main theorem.
\begin{proof}[Proof of of \prettyref{thm:Main-thm}]
 We employ Duhamel's formula for $t\mapsto e^{i\mathbb{H}t}Re^{-i\mathbb{H}^{\text{eff}}t}$:

\begin{align}
\|R^{\ast}e^{-i\mathbb{H}t}R\psi-e^{-i\mathbb{H}^{\text{eff}}t}\psi\| & =\|R\psi-e^{i\mathbb{H}t}Re^{-i\mathbb{H}^{\text{eff}}t}\psi\|\\
 & =\|\int_{0}^{t}\d s\ e^{i\mathbb{H}s}\big(\mathbb{H}R-R\mathbb{H}^{\text{eff}}\big)e^{-i\mathbb{H}^{\text{eff}}s}\psi\|\\
 & \leq\int_{0}^{t}\d s\ \|(R^{\ast}\mathbb{H}R-\mathbb{H}^{\text{eff}})e^{-i\mathbb{H}^{\text{eff}}s}\psi\|
\end{align}
where we used that $R$ is unitary.

From our considerations in \eqref{eq:ph-transformed-H} it follows
that 

\begin{align}
R^{\ast}\mathbb{H}R-\mathbb{H}^{\textnormal{eff}} & =-\beta\Delta_{y}+\mathbb{H}_{0}+b^{\ast}(\tilde{h}_{y})+b(\tilde{h}_{y})+E_{N}^{\text{pw}}+\mathcal{E}-\mathbb{H}^{\textnormal{eff}}\nonumber \\
 & =\mathbb{H}_{0}-\mathbb{D}_{B}+b^{\ast}(\tilde{h}_{y})-c^{\ast}(h_{y})+b(\tilde{h}_{y})-c(h_{y})+\mathcal{E}.
\end{align}
The interaction terms can be approximated in the sense of \prettyref{lem:Approx-of-patch-operators}:
\begin{align*}
\|\left(b^{\ast}(\tilde{h}_{y})-c^{\ast}(h_{y})\right)\psi\| & \leq\lambda\sum_{\kinG}|\hat{V}(k)|\ \|\big(b(k)-\sumak n_{\alpha}(k)\cak\big)\psi\|\\
 & \leq C\lambda N^{\frac{1}{3}}\|\hat{V}\|_{1}(N^{-\frac{\delta}{2}}+N^{-\frac{1}{6}}M^{\frac{1}{4}})\|(\mathcal{N}+1)^{\frac{1}{2}}\psi\|.
\end{align*}

Therefore by combining the bounds from \prettyref{prop:Approx-kinetic-energy}
with $(\mathbb{H}_{0}-\mathbb{D}_{B})\psi=0$, \prettyref{lem:Estimate-non-bosonizable}
with $\|\hat{V}\|_{1}<C$ and \prettyref{prop:Effective-time-evolution-number-op}
we obtain
\begin{align}
 & \|R^{\ast}e^{-i\mathbb{H}t}R\psi-e^{-i\mathbb{H}^{\text{eff}}t}\psi\|\nonumber \\
 & \leq\int_{0}^{t}\d s\ \|(\mathbb{H}_{0}-\mathbb{D}_{B})e^{-i\mathbb{H}^{\text{eff}}s}\psi\|+\int_{0}^{t}\d s\ \|\mathcal{E}e^{-i\mathbb{H}^{\text{eff}}s}\psi\|\nonumber \\
 & \quad+\int_{0}^{t}\d s\ \|\left(b^{\ast}(\tilde{h}_{y})-c^{\ast}(h_{y})\right)e^{-i\mathbb{H}^{\text{eff}}s}\psi\|+\int_{0}^{t}\d s\ \|\left(b(\tilde{h}_{y})-c(h_{y})\right)e^{-i\mathbb{H}^{\text{eff}}s}\psi\|\nonumber \\
 & \leq C(1+\lambda^{-1}) N^{\frac{1}{3}}\left(M^{-\frac{1}{2}}\|(\mathcal{N}+3)\psi\|+\kF MN^{-1+\delta}\|(\mathcal{N}+3)^{2}\psi\|\right)\int_{0}^{t}\d s\left(e^{C\lambda\kF t}-1\right)\nonumber \\
 & \quad+C\lambda\int_{0}^{t}\d s\ \|\mathcal{N}e^{-i\mathbb{H}^{\text{eff}}s}\psi\|\nonumber \\
 & \quad+C\lambda N^{\frac{1}{3}}(N^{-\frac{\delta}{2}}+N^{-\frac{1}{6}}M^{\frac{1}{4}})\int_{0}^{t}\d s\ \|(\mathcal{N}+1)^{\frac{1}{2}}e^{-i\mathbb{H}^{\text{eff}}s}\psi\|\nonumber \\
 & \leq C(1+\lambda^{-1})\lambda^{-1}N^{\frac{1}{3}}\left(\kF^{-1}M^{-\frac{1}{2}}\|(\mathcal{N}+3)\psi\|+MN^{-1+\delta}\|(\mathcal{N}+3)^{2}\psi\|\right)\left(e^{C\lambda\kF t}-\lambda\kF t-1\right)\nonumber \\
 & \quad+C\kF^{-1}\|(\mathcal{N}+3)\psi\|\left(e^{C\lambda\kF t}-1\right)\nonumber \\
 & \quad+C (N^{-\frac{\delta}{2}}+N^{-\frac{1}{6}}M^{\frac{1}{4}})\|(\mathcal{N}+3)^{\frac{1}{2}}\psi\|\left(e^{C\lambda\kF t}-1\right)\nonumber \\
 & \leq C\tilde{C}\|(\mathcal{N}+3)^{2}\psi\|\left(e^{C\lambda\kF t}-1\right).
\end{align}

The prefactor is given by
\begin{equation}
\tilde{C} =\max\left\{ (1+\lambda^{-1})\lambda^{-1}\kF^{-1}N^{\frac{1}{3}}\left(M^{-\frac{1}{2}}+\kF MN^{-1+\delta}\right),\kF^{-1}, (N^{-\frac{\delta}{2}}+N^{-\frac{1}{6}}M^{\frac{1}{4}})\right\}
\end{equation}
where we optimize over $M$ and $\delta$ with  $\lambda\geq \kF^{-\frac{1}{6}}$
to obtain the desired result. 
\end{proof}

\section{\label{sec:Coheren-Proof}Proof of \prettyref{thm:Effective-coherent-dynamics} }

\paragraph*{Properties of almost-bosonic coherent states} We first show the following rigorous properties of the Weyl operator which was introduced in \prettyref{def:Weyl-op}:
\begin{lem}[Approximate shift property]
\label{lem:Approximate-shift}Let $\eta,\xi\in\bigoplus_{\kinG}l^{2}(\mathcal{I}_{k})$
and $W_{\sigma}(\eta)\coloneqq e^{\sigma B}\coloneqq e^{\sigma c^{\ast}(\eta)-\sigma c(\eta)}$
for all $\sigma\in[0,1]$, then it holds 
\begin{align*}
W_{\sigma}(\eta)^{\ast}c(\xi)W_{\sigma}(\eta) & =c(\xi)+\sigma\langle\xi,\eta\rangle+\langle\xi,\mr{\sigma}\rangleka,\\
W_{\sigma}(\eta)^{\ast}c^{\ast}(\xi)W_{\sigma}(\eta) & =c^{\ast}(\xi)+\sigma\langle\eta,\xi\rangle+\langle\mr{\sigma},\xi\rangleka,
\end{align*}
with
$\langle\cdot, \cdot \rangleka: \bigoplus_{\kinG}l^{2}(\mathcal{I}_{k}) \times \bigoplus_{\kinG}l^{2}(\mathcal{I}_{k}) \to \mathbb{C}$
given by
\begin{equation}
\langle\xi,\mr{\sigma}\rangleka\coloneqq\sum_{\kinG}\sumak\xi_{\alpha}(k)^{\ast}\mr{\sigma}_{\alpha}(k)
\end{equation}
and a $\sigma$-dependent error term $\mr{\sigma}_{\alpha}(k)\coloneqq\int_{0}^{\sigma}\d\tau\ e^{-\tau B}\left(\sum_{\linG}\eta_{\alpha}(l)\mathcal{E}_{\alpha}(l,k)\right)e^{\tau B}$.
\end{lem}

\begin{rem}
We will give an estimate for $\mathcal{R}$ to show that this
term corresponds indeed to a small error. Note that it holds $\mathcal{E}_{\alpha}(k,l)=\mathcal{E}_{\alpha}(l,k)^{\ast}$
for all $l,\kinG$ and $\alpha\in\mathcal{I}_{k}\cap\mathcal{I}_{l}$
from \prettyref{lem:CAR-error} and therefore 
\begin{align}
\mr{\sigma}_{\gamma}(l)^{\ast}\xi_{\gamma}(l) & =\int_{0}^{\sigma}\d\tau\ e^{-\tau B}\left(\sum_{\kinG}\mathcal{E}_{\gamma}(l,k)\overline{\eta_{\gamma}(k)}\xi_{\gamma}(l)\right)e^{\tau B},\\
\overline{\xi_{\gamma}(l)}\mr{\sigma}_{\gamma}(l) & =\int_{0}^{\sigma}\d\tau\ e^{-\tau B}\left(\sum_{\kinG}\overline{\xi_{\gamma}(l)}\eta_{\gamma}(k)\mathcal{E}_{\gamma}(k,l)\right)e^{\tau B}.
\end{align}

For $\xi=\eta$ the above equations coincide, i.e.
\begin{equation}
\langle\eta,\mr{\sigma}\rangleka=\langle\mr{\sigma},\eta\rangleka
\end{equation}
from which it follows immediately that $\big(c^{\ast}(\eta)-c(\eta)\big)W_{\sigma}(\eta)=W_{\sigma}(\eta)\big(c^{\ast}(\eta)-c(\eta)\big)$,
i.e. $[B,W_{\sigma}(\eta)]=0$.
\end{rem}

\begin{proof}
We observe that  $W_{\sigma}(\eta)= e^{\sigma B}$
defines a strongly continuous one-parameter semigroup. Thus we can 
define a derivative and make use of Duhamel's formula of the form
\begin{equation}
  e^{-\sigma B}\cgl e^{\sigma B}  = \cgl + \int_{0}^{\sigma}\d\tau\ e^{-\tau B}[\cgl,B]e^{\tau B}.  \\
\end{equation}
The interested reader is referred to \cite{Pazy1983,Engel2006} where definitions and properties of operator 
derivatives are discussed.
The desired statement follows with the CCR
as stated in \eqref{eq:approx-CAR}
\begin{align}
[\cgl,e^{\sigma B}] & =e^{\sigma B}\int_{0}^{\sigma}\d\tau\ e^{-\tau B}[\cgl,B]e^{\tau B}\nonumber \\
 & =e^{\sigma B}\int_{0}^{\sigma}\d\tau\ e^{-\tau B}\left(\eta_{\gamma}(l)+\sum_{\kinG}\eta_{\gamma}(k)\mathcal{E}_{\gamma}(k,l)\right)e^{\tau B}\nonumber \\
 & =\sigma\eta_{\gamma}(l)e^{\sigma B}+e^{\sigma B}\mr{\sigma}_{\gamma}(l)
\end{align}
Since $c(\xi)\equiv\sum_{\kinG}\sumak\overline{\xi_{\alpha}(k)}\cak$
is linear the result follows from the above identity.
\end{proof}
The following statement shows that the number of particles in the state $W(\eta)\phi\otimes\Omega$ corresponds to a random variable with expectation approximately
being $2\|\eta\|^{2}$.
\begin{prop}[Expectation of the number operator]
\label{prop:expectation-number-W} Let $\zeta=\phi\otimes\Omega\in\domH$,
then it holds for all $\eta\in\bigoplus_{\kinG}l^{2}(\mathcal{I}_{k})$
\[
\langle W(\eta)\zeta,\mathcal{N}W(\eta)\zeta\rangle=2\|\eta\|^{2}+4\int_{0}^{1}\d\sigma\langle\zeta,\langle\eta,\mr{\sigma}\rangleka\zeta\rangle.
\]
\end{prop}

\begin{proof}
Using $W^{\ast}W=\id$ yields for all $\zeta\in\domH$ with $\|\zeta\|=1$
\begin{equation}
\langle W(\eta)\zeta,\mathcal{N}W(\eta)\zeta\rangle=\langle W(\eta)\zeta,[\mathcal{N},W(\eta)]\zeta\rangle+\langle\zeta,\mathcal{N}\zeta\rangle.\label{eq:coherent-exp-1}
\end{equation}
We use Duhamel's formula to calculate 
\begin{align}
[\mathcal{N},e^{B}] & =e^{B}\int_{0}^{1}\d\tau\ e^{-\tau B}[\mathcal{N},B]e^{\tau B}\nonumber \\
 & =2e^{B}\int_{0}^{1}\d\tau\ e^{-\tau B}\left(c^{\ast}(\eta)+c(\eta)\right)e^{\tau B}\nonumber \\
 & =2e^{B}\int_{0}^{1}\d\tau\ e^{-\tau B}\left(B+2c(\eta)\right)e^{\tau B}\nonumber \\
 & =2e^{B}B+4e^{B}\int_{0}^{1}\d\tau\ e^{-\tau B}c(\eta)e^{\tau B}\label{eq:coherent-exp-2}
\end{align}
where we used \eqref{eq:number-op-vs-ccak}. Therefore
\begin{align}
\langle W(\eta)\zeta,[\mathcal{N},W(\eta)]\zeta\rangle & =2\langle\zeta,B\zeta\rangle+4\int_{0}^{1}\d\tau\ \langle e^{\tau B}\zeta,c(\eta)e^{\tau B}\zeta\rangle\nonumber \\
 & =2\langle\zeta,B\zeta\rangle+2\|\eta\|^{2}+4\langle\zeta,c(\eta)\zeta\rangle+4\int_{0}^{1}\d\tau\langle\zeta,\langle\eta,\mr{\tau}\rangleka\zeta\rangle\label{eq:coherent-exp-3}
\end{align}
where we used the shift property \prettyref{lem:Approximate-shift}
and that $e^{\tau B}$ is unitary
\begin{align}
\langle e^{\tau B}\zeta,c(\eta)e^{\tau B}\zeta\rangle & =\langle e^{\tau B}\zeta,[c(\eta),e^{\tau B}]\zeta\rangle+\langle\zeta,c(\eta)\zeta\rangle\nonumber \\
 & =\tau\|\eta\|^{2}+\langle\zeta,\langle\eta,\mr{1}\rangleka\zeta\rangle+\langle\zeta,c(\eta)\zeta\rangle.
\end{align}
Inserting \eqref{eq:coherent-exp-3} and \eqref{eq:coherent-exp-2}
into \eqref{eq:coherent-exp-1} we obtain
\begin{equation}
\langle W(\eta)\zeta,\mathcal{N}W(\eta)\zeta\rangle=2\|\eta\|^{2}+2\langle\zeta,\left(c^{\ast}(\eta)+c(\eta)\right)\zeta\rangle+\langle\zeta,\mathcal{N}\zeta\rangle+4\int_{0}^{1}\d\tau\langle\zeta,\langle\eta,\mr{\tau}\rangleka\zeta\rangle.
\end{equation}
The desired result holds for $\zeta=\phi\otimes\Omega$ since
$c(\eta)\phi\otimes\Omega=0$.

\end{proof}
For later purposes, we can bound the expectation of the number operator
in the following way:
\begin{prop}[Stability of the number operator]
\label{prop:Stability-number-op}Let $\eta\in\bigoplus_{\kinG}l^{2}(\mathcal{I}_{k})$.
There exists a constant $C>0$ such that it holds for all $\tau\in[-1,1],n\in\mathbb{N}$
and $\zeta\in\domH$
\[
\langle e^{\tau B}\zeta,(\mathcal{N}+1)^{n}e^{\tau B}\zeta\rangle\leq e^{C\|\eta\|n|\tau|}\langle\zeta,(\mathcal{N}+3)^{n}\zeta\rangle.
\]
\end{prop}

\begin{proof}
The proof works analogously to the proof of \prettyref{prop:Effective-time-evolution-number-op}
with a Grönwall argument and $B$ instead of $\Heff$ which is given
later.  Note that 
\begin{equation}
[\mathcal{N},B]=[\mathcal{N},c^{\ast}(\eta)-c(\eta)]=2\sum_{\kinG}\sumak\eta_{\alpha}(k)\big(\ccak+\cak\big)
\end{equation}
where we again used \eqref{eq:number-op-vs-ccak}. The result is than
obtained by using the same estimates with $\|\eta\|$ taking the role
of $\|h_{y}\|$.
\end{proof}
\begin{lem}
\label{lem:derivative-of-Weyl-op}Let $\eta_{t}\in\bigoplus_{\kinG}l^{2}(\mathcal{I}_{k})$
be differentiable in $t$ with derivative $\dot{\eta_{t}}\in\bigoplus_{\kinG}l^{2}(\mathcal{I}_{k})$
for all $t\in\R$. Then it holds for all $t\in\R$ 
\[
\partial_{t}W(\eta_{t})=\left(c^{\ast}(\dot{\eta}_{t})-c(\dot{\eta}_{t})+i\textnormal{Im}\langle\dot{\eta}_{t},\eta_{t}\rangle\right)W(\eta_{t})+2i\int_{0}^{1}\d\tau\ W_{(1-\tau)}(\eta_{t})\textnormal{Im}\langle\dot{\eta}_{t},\mr{1-\tau}\rangleka W_{\tau}(\eta_{t})
\]
with the shorthand notation $\textnormal{Im}\langle A,B\rangleka\coloneqq-\frac{i}{2}\sum_{\kinG}\sumak(A_{\alpha}^{\ast}(k)B_{\alpha}(k)-B_{\alpha}^{\ast}(k)A_{\alpha}(k))$
.
\end{lem}

\begin{proof}
For arbitrary $s\in\R$ it holds
\begin{align}
W(\eta_{s})^{\ast}\partial_{s}W(\eta_{s}) & =\left. e^{-\tau B_{s}}\partial_{s}e^{\tau B_{s}}\right|_{\tau=0}^{\tau=1}=\int_{0}^{1}\d\tau\ \partial_{\tau}\left(e^{-\tau B_{s}}\partial_{s}e^{\tau B_{s}}\right)\nonumber \\
 & =\int_{0}^{1}\d\tau\left( -B_{s}e^{-\tau B_{s}}\partial_{s}e^{\tau B_{s}}+e^{-\tau B_{s}}\partial_{s}\partial_{\tau}e^{\tau B_{s}}\right) \\
 & =\int_{0}^{1}\d\tau\left( -B_{s}e^{-\tau B_{s}}\partial_{s}e^{\tau B_{s}}+e^{-\tau B_{s}}\partial_{s}\left(B_{s}e^{\tau B_{s}}\right)\right) \\
 & =\int_{0}^{1}\d\tau\left( -B_{s}e^{-\tau B_{s}}\partial_{s}e^{\tau B_{s}}+e^{-\tau B_{s}}\left(\partial_{s}B_{s}\right)e^{\tau B_{s}}+e^{-\tau B_{s}}B_{s}\partial_{s}e^{\tau B_{s}}\right) \\
 & =\int_{0}^{1}\d\tau\ e^{-\tau B_{s}}\left(\partial_{s}B_{s}\right)e^{\tau B_{s}}.
\end{align}
Thus 
\begin{align}
\partial_{s}W(\eta_{s}) & =\int_{0}^{1}\d\tau\ e^{(1-\tau)B_{s}}\left(\partial_{s}B_{s}\right)e^{\tau B_{s}}\nonumber \\
 & =\int_{0}^{1}\d\tau\ \left(\partial_{s}B_{s}\right)e^{(1-\tau)B_{s}}e^{\tau B_{s}}+\int_{0}^{1}\d\tau\ \left[e^{(1-\tau)B_{s}},\partial_{s}B_{s}\right]e^{\tau B_{s}}\\
 & =\left(\partial_{s}B_{s}\right)e^{B_{s}}+\int_{0}^{1}\d\tau\ \left[e^{(1-\tau)B_{s}},\partial_{s}B_{s}\right]e^{\tau B_{s}}.
\end{align}
With 
\begin{equation}
\partial_{s}B_{s}=\partial_{s}\left\{ c^{\ast}(\eta_{s})-c(\eta_{s})\right\} =c^{\ast}(\dot{\eta}_{s})-c(\dot{\eta}_{s})
\end{equation}
from the linearity of $c(\eta_{s})$ it follows
\begin{align}
 & \left[e^{(1-\tau)B_{s}},\partial_{s}B_{s}\right]e^{\tau B_{s}}\nonumber \\
 & =\left( [W_{(1-\tau)}(\eta_{s}),c^{\ast}(\dot{\eta}_{s})]-[W_{(1-\tau)}(\eta_{s}),c(\dot{\eta}_{s})]\right) W_{\tau}(\eta_{s})\\
 & =W_{(1-\tau)}(\eta_{s})\bigg((1-\tau)\langle\dot{\eta}_{s},\eta_{s}\rangle-(1-\tau)\langle\eta_{s},\dot{\eta}_{s}\rangle\\
 & \hphantom{=W_{(1-\tau)}(\eta_{s})\bigg\{(1-\tau)\langle\dot{\eta}_{s},\eta_{s}\rangle}+\langle\dot{\eta}_{s},\mr{1-\tau}\rangleka-\langle\mr{1-\tau},\dot{\eta}_{s}\rangleka\bigg) W_{\tau}(\eta_{s})\\
 & =W(\eta_{s})(1-\tau)2i\text{Im}\langle\dot{\eta}_{s},\eta_{s}\rangle+2iW_{(1-\tau)}(\eta_{s})\text{Im}\langle\dot{\eta}_{s},\mr{1-\tau}\rangleka W_{\tau}(\eta_{s}).
\end{align}
Inserting the above identity yields
\begin{equation}
\partial_{s}W(\eta_{s})=\left(c^{\ast}(\dot{\eta}_{s})-c(\dot{\eta}_{s})+i\text{Im}\langle\dot{\eta}_{s},\eta_{s}\rangle\right)W(\eta_{s})+2i\int_{0}^{1}\d\tau\ W_{(1-\tau)}(\eta_{s})\text{Im}\langle\dot{\eta}_{s},\mr{1-\tau}\rangleka W_{\tau}(\eta_{s}).
\end{equation}
\end{proof}

\paragraph*{Proof of the main theorem} First, we collect some useful observations on the function $\eta_{s}$
in the form of the following two lemmata. We postpone the proofs
to the end of the section in order to concentrate on presenting the proof of the main result.

\begin{lem}
\label{lem:eta-bounds}Let $\eta_{s}$ be defined as in \eqref{eq:Coherent-state-variableI}
for $N^{2\delta}\ll M\ll N^{\frac{2}{3}-2\delta}$, then it holds
for all $s\in\R$
\begin{align*}
\|\eta_{s}\|^{2} & =\pi\lambda^{2}\sum_{\kinG}\frac{\hat{V}(k)^{2}}{|k|}\left( \log(2\kF|k|s)-\textnormal{Ci}(2\kF|k|s)+\gamma\right) \\
 & \hphantom{\frac{\pi\lambda^{2}}{4^{2}}\sum_{\kinG}\frac{\hat{V}(k)^{2}}{|k|}}\times\left\{ 1+g(\kF|k|s)\ \mathcal{O}\left(M^{\frac{1}{2}}N^{-\frac{1}{3}+\delta}+N^{-\delta}\right)\right\} .
\end{align*}
where $\gamma$ is the Euler-Mascheroni constant and $g:\R\to\R_{\geq0}$ is a function independent of $\kF$ and
monotonically increasing.

Furthermore, define $f:\R\times\R\to\R$ by 
\begin{equation*}
 (y,x)\mapsto f_{y}(x)\coloneqq\min\{e^{\sqrt{\pi}\|\hat{V}(\cdot)^{1/2}\|_{2} yx}, e^{\sqrt{2\pi}\|\hat{V}\|_{2} (\log(18)+\frac{1}{9}) y} e^{\frac{\sqrt{8\pi}}{9}\|\hat{V}\|_{2} yx}\}.
\end{equation*}
Then there exists a $C>0$ independent of $\kF$ such that for $c_{0}>0$
and $\kF$ sufficiently large 
\begin{align}
\|\eta_{s}\| & \leq  \log\left(f_{1}( \lambda\kF s)\right), \label{eq:eta-bound} \\
e^{c_{0}\|\eta_{s}\|} & \leq  f_{c_{0}  }( \lambda \kF s). \label{eq:exp-eta-bound}
\end{align}
\end{lem}

\begin{rem}
Note that $f_{y}$ is for all $y\geq0$ monotonically increasing with
$f_{y}(0)=1$. 
\end{rem}

The previous statement is useful when combined with the following estimate:

\begin{lem} \label{lem:eta-bounds2}
There exists a constant $C>0$ only depending on $V$ such that it holds for all $s\in\R$, $n\in\mathbb{N}$, $\psi\in\mathcal{F}$
\begin{enumerate}
\item $\sum_{\kinG}\||k|^{n}\eta_{s}(k)\|_{l^{2}}\leq C\|\eta_{s}\|$,
\item $\langle\eta_{s},|k|^{n}\eta_{s}\rangle\leq C\|\eta_{s}\|^{2}$,
\item $\|c^{\ast}(|k|^{n}\eta_{s})\psi\|\leq C\|\eta_{s}\|\ \|(\mathcal{N}+1)^{1/2}\psi\|$.
\end{enumerate}
\end{lem}

We will now give the proof of the second main theorem.
\begin{proof}[Proof of \prettyref{thm:Effective-coherent-dynamics}]
 We use the approach as sketched in \prettyref{rem:Heuristic-calculation}.
Since the bosonic property holds only with an error, the equality
\eqref{eq:exactly-bosonic-equality} holds only approximately:
\begin{align}
\|e^{-i\mathbb{H}^{\text{eff}}t}\psi-e^{iP(t)}W(\eta_{t})\psi\| 
& =\|\psi-e^{i\mathbb{H}^{\text{eff}}t}e^{-iE_{N}^{\text{pw}}t}e^{i2\text{Im}(\nu_{t})}e^{-i\text{Im}\int_{0}^{t}\d s\langle\dot{\eta}_{s},\eta_{s}\rangle}W(\eta_{t})\psi\| \nonumber \\ 
& \leq\int_{0}^{t}\d s\|\big(\mathbb{H}^{\text{eff}}-E_{N}^{\text{pw}}+2\text{Im}(\dot{\nu}_{s})-\text{Im}\langle\dot{\eta}_{s},\eta_{s}\rangle\big)W(\eta_{s})\psi-i\partial_{s}W(\eta_{s})\psi\| \nonumber \\
& \leq \int_{0}^{t}\d s\|h_{0}W(\eta_{s})\phi\otimes\Omega\|+\| \text{Error}_{1}\|+\| \text{Error}_{2}\|.\label{eq:almost-bosonic-inequality}
\end{align}

We will first estimate the error terms and then subsequently treat
the $h_{0}$ term in a separate lemma. We give an explicit expression
for the first error term using \prettyref{lem:Approximate-shift}
on the approximate shift property applied to $\cak$ in the $c^{\ast}c(\epsilon),c(h_{y})$
and $c(i\dot{\eta}_{s})$ terms. Thus, the error of \eqref{eq:exactly-bosonic-equality}
is given by 
\begin{equation}
\text{Error}_{1}\coloneqq\int_{0}^{t}\d s\sum_{\kinG}\sumak\left( \epsilon_{\alpha}(k)\ccak+(1-e^{is\epsilon_{\alpha}(k)})\overline{\phiak}\right) W(\eta_{s})\mr{1}_{\alpha}(k)\psi.
\end{equation}
The second error term is given by \prettyref{lem:derivative-of-Weyl-op}
on the time derivative of the almost-bosonic Weyl operator and therefore
\begin{align}
\text{Error}_{2}\coloneqq- & 2i\int_{0}^{t}\d s\int_{0}^{1}\d\tau\ W_{1-\tau}(\eta_{s})\text{Im}\langle\dot{\eta}_{s},\mr{1-\tau}\rangleka W_{\tau}(\eta_{s})\psi\nonumber \\
\equiv & 2\int_{0}^{t}\d s\int_{0}^{1}\d\tau e^{(1-\tau)B}\bigg(\sum_{\kinG}\sumak e^{is\epsilon_{\alpha}(k)}\overline{\phiak}\mr{1-\tau}_{\alpha}(k)\nonumber \\
 & \hphantom{2\int_{0}^{t}\d s\int_{0}^{1}\d\tau}-\sum_{\kinG}\sumak\mr{1-\tau}_{\alpha}(k)^{\ast}e^{-is\epsilon_{\alpha}(k)}\phiak\bigg) e^{\tau B}\psi.
\end{align}
Firstly, we show that the term $\mr{\lambda}_{\alpha}(k)\psi=\int_{0}^{\sigma}\d\tau\ e^{-\tau B}\left(\sum_{\linG}\eta_{\alpha}(l)\mathcal{E}_{\alpha}(l,k)\right)e^{\tau B}\psi$
as defined in \prettyref{lem:Approximate-shift} constitutes indeed
a small error.  We estimate 
\begin{align}
 & \sum_{\linG}\sum_{\gamma\in\mathcal{I}_{l}}\|\mr{1}_{\gamma}(l)\psi\|^{2}\nonumber \\
 & \leq\sum_{\linG}\sum_{\gamma\in\mathcal{I}_{l}\cap\mathcal{I}_{k}}\bigg(\sum_{\kinG}|\eta_{\gamma}(k)|\int_{0}^{1}\d\tau\|\mathcal{E}_{\gamma}(k,l)e^{\tau B}\psi\|\bigg)^{2}\nonumber \\
 & \leq C\left(MN^{-\frac{2}{3}+\delta}(e^{C\|\eta_{s}\|}-1)\|(\mathcal{N}+3)\psi\|\right)^{2}\label{eq:coherent-state-error-I}
\end{align}
where we used $\sumak\bigg(\sum_{\kinG}|\eta_{\alpha}(k)|\bigg)^{2}\leq\sum_{k',\kinG}\|\eta(k)\|_{l^{2}}\|\eta(k')\|_{l^{2}}\leq C\|\eta\|^{2}$
by \prettyref{lem:eta-bounds} and 
\begin{align}
 & \int_{0}^{1}\d\tau\|\mathcal{E}_{\gamma}(k,l)e^{\tau B}\psi\|\\
 & \leq\int_{0}^{1}\d\tau\langle e^{\tau B}\psi,|\mathcal{E}_{\gamma}(k,l)|^{2}e^{\tau B}\psi\rangle^{1/2}\leq CMN^{-\frac{2}{3}+\delta}\int_{0}^{1}\d\tau\langle e^{\tau B}\psi,\mathcal{N}^{2}e^{\tau B}\psi\rangle^{1/2}\\
 & \leq CMN^{-\frac{2}{3}+\delta}\int_{0}^{1}\d\tau e^{C\|\eta_{s}\|\tau}\|(\mathcal{N}+3)\psi\|\leq C\|\eta_{s}\|^{-1}(e^{C\|\eta_{s}\|}-1)MN^{-\frac{2}{3}+\delta}\|(\mathcal{N}+3)\psi\|
\end{align}
which follows from $e^{\tau B}$ is unitary in the first inequality,
\prettyref{lem:CAR-error} in the second inequality and \prettyref{prop:Stability-number-op}
in the third inequality. 

Secondly, we estimate
\begin{align}
 & \sum_{\kinG}\sumak\|\ccak e^{B}\mr{1}_{\alpha}(k)\psi\|\nonumber \\
 & \leq\sum_{\kinG}\sumak\int_{0}^{1}\d\tau\|\ccak e^{(1-\tau)B}\left(\sum_{\linG}\eta_{\alpha}(l)\mathcal{E}_{\alpha}(l,k)\right)e^{\tau B}\psi\|\nonumber \\
 & \leq\int_{0}^{1}\d\tau\sum_{k,\linG}\sum_{\alpha\in\mathcal{I}_{k}\cap\mathcal{I}_{l}}|\eta_{\alpha}(l)|\ \|\ccak\mathcal{E}_{\alpha}(l,k)e^{\tau B}\psi\|\nonumber \\
 & \quad+\int_{0}^{1}\d\tau\sum_{k,\linG}\sum_{\alpha\in\mathcal{I}_{k}\cap\mathcal{I}_{l}}|\eta_{\alpha}(l)|\ \|[\ccak,e^{(1-\tau)B}]\mathcal{E}_{\alpha}(l,k)e^{\tau B}\psi\|\nonumber \\
 & \leq C\|\eta_{s}\|MN^{-\frac{2}{3}+\delta}\int_{0}^{1}\d\tau\bigg(\|(\mathcal{N}+1)^{\frac{3}{2}}e^{\tau B}\psi\|^{2}\bigg)^{1/2}+C\|\eta_{s}\|^{2}MN^{-\frac{2}{3}+\delta}\int_{0}^{1}\d\tau\ (1-\tau)\|\mathcal{N}e^{\tau B}\psi\|\nonumber \\
 & \quad+\sum_{\kinG}\|\eta_{s}\|\int_{0}^{1}\d\tau\bigg(\sum_{\linG}\sum_{\alpha\in\mathcal{I}_{k}\cap\mathcal{I}_{l}}\|\mr{1}_{\alpha}(k)\mathcal{E}_{\alpha}(l,k)e^{\tau B}\psi\|^{2}\bigg)^{1/2}\nonumber \\
 & \leq CMN^{-\frac{2}{3}+\delta}(e^{C\|\eta_{s}\|}-1)\|(\mathcal{N}+3)^{\frac{3}{2}}\psi\|+CMN^{-\frac{2}{3}+\delta}(e^{C\|\eta_{s}\|}-C\|\eta_{s}\|-1)\|(\mathcal{N}+3)\psi\|\nonumber \\
 & \quad+CM^{\frac{3}{2}}N^{-\frac{2}{3}+\delta}(e^{C\|\eta_{s}\|}-1)^{2}\|(\mathcal{N}+3)^{2}\psi\|\nonumber \\
 & \leq CM^{\frac{3}{2}}N^{-\frac{2}{3}+\delta}(e^{C\|\eta_{s}\|}-1)\|(\mathcal{N}+3)^{2}\psi\|.\label{eq:coherent-state-error-II}
\end{align}
where we used $[\ccak,e^{\sigma B}]=-\lambda\eta_{\alpha}(k)e^{\sigma B}-e^{\sigma B}\mr{1}_{\alpha}(k)$
from \prettyref{lem:Approximate-shift}, the Cauchy-Schwarz inequality
for the $\alpha$-summation, \prettyref{lem:CAR-error} in the third
inequality and in the fourth inequality we used \prettyref{prop:Stability-number-op}
and \eqref{eq:coherent-state-error-I}.

In total by combining \eqref{eq:coherent-state-error-I} and \eqref{eq:coherent-state-error-II}
we end up with the following estimate 
\begin{align}
\|\text{Error}_{1}\| & \leq\int_{0}^{t}\d s\sum_{\kinG}\sumak\|\left\{ \epsilon_{\alpha}(k)\ccak+(1-e^{is\epsilon_{\alpha}(k)})\overline{\phiak}\right\} e^{B}\mr{1}_{\alpha}(k)\psi\|\nonumber \\
 & \leq C\kF\int_{0}^{t}\d s\sum_{\kinG}\sumak\|\ccak e^{B}\mr{1}_{\alpha}(k)\psi\|\nonumber \\
 & \quad+C\lambda\int_{0}^{t}\d s\sum_{\kinG}|\hat{V}(k)|\sumak\|n_{\alpha}(k)e^{B}\mr{1}_{\alpha}(k)\psi\|\nonumber \\
 & \leq C\kF\int_{0}^{t}\d s\sum_{\kinG}\sumak\|\ccak e^{B}\mr{1}_{\alpha}(k)\psi\|\nonumber \\
 & \quad+C\lambda\kF\int_{0}^{t}\d s\|\hat{V}\|_{2}\bigg(\sum_{\kinG}\sumak\|\mr{1}_{\alpha}(k)\psi\|^{2}\bigg)^{1/2}\nonumber \\
 & \leq C\kF MN^{-\frac{2}{3}+\delta}\int_{0}^{t}\d s\ (e^{C\|\eta_{s}\|}-1)\|(\mathcal{N}+3)^{2}\psi\|\nonumber \\
 & \quad+C\lambda\kF MN^{-\frac{2}{3}+\delta}\int_{0}^{t}\d s\ (e^{C\|\eta_{s}\|}-1)\|(\mathcal{N}+3)\psi\|\nonumber \\
 & \leq C\ (f_{C}( \lambda\kF t)-1)(\lambda+1)\kF tMN^{-\frac{2}{3}+\delta}\|(\mathcal{N}+3)^{2}\psi\|.\label{eq:Error-estimate}
\end{align}

where we used \eqref{lem:Approx-n_alpha(k)^2} and $e^{B}$ unitary
in the third inequality and \prettyref{lem:eta-bounds} in the last
line. Using $\psi=\phi\otimes\Omega$
we obtain the desired bound. 

Similarly, we obtain an estimate for the second error term using Cauchy-Schwarz,
\eqref{eq:coherent-state-error-I} and \prettyref{prop:Stability-number-op}
\begin{align}
 & \|\text{Error}_{2}\|\nonumber \\
&\leq  2\int_{0}^{t}\d s\int_{0}^{1}\d\tau\sum_{\kinG}\sumak\bigg(\|\overline{\phiak}\mr{1-\tau}_{\alpha}(k)e^{\tau B}\psi\|+\|\mr{1-\tau}_{\alpha}(k)^{\ast}\phiak e^{\tau B}\psi\|\bigg)\nonumber \\
&\leq  2\lambda\int_{0}^{t}\d s\int_{0}^{1}\d\tau\sum_{\kinG}\sumak|\hat{V}(k)n_{\alpha}(k)|\bigg(\|\mr{1-\tau}_{\alpha}(k)e^{\tau B}\psi\|+\|\mr{1-\tau}_{\alpha}(k)^{\ast}e^{\tau B}\psi\|\bigg)\nonumber \\
&\leq  C\int_{0}^{t}\d s\int_{0}^{1}\d\tau\ \lambda\kF\|\hat{V}\|_{2}\bigg(\bigg(\sum_{\kinG}\sumak\|\mr{1-\tau}_{\alpha}(k)e^{\tau B}\psi\|^{2}\bigg)^{1/2}+\bigg(\sum_{\kinG}\sumak\|\mr{1-\tau}_{\alpha}(k)^{\ast}e^{\tau B}\psi\|^{2}\bigg)^{1/2}\bigg)\nonumber \\
&\leq  C\lambda\kF MN^{-\frac{2}{3}+\delta}\int_{0}^{t}\d s\int_{0}^{1}\d\tau\ (e^{C\|\eta_{s}\|(1-\tau)}-1)\|(\mathcal{N}+3)e^{\tau B}\psi\|\nonumber \\
&\leq  C\lambda\kF MN^{-\frac{2}{3}+\delta}\int_{0}^{t}\d s\int_{0}^{1}\d\tau\ (e^{C\|\eta_{s}\|}-e^{C\|\eta_{s}\|\tau})\|(\mathcal{N}+5)\psi\|\nonumber \\
&\leq  C(f_{C}( \lambda\kF t)-1)\lambda\kF tMN^{-\frac{2}{3}+\delta}\|(\mathcal{N}+5)\psi\|.\label{eq:Error2-estimate}
\end{align}
Again, by using $\psi=\phi\otimes\Omega$
we obtain the desired bound. 

With the subsequent \prettyref{lem:Laplacian-estimate}, we can conclude
with a bound on $h_{0}=-\beta \Delta_y $ of the form
\begin{align}
 & \int_{0}^{t}\d s\|h_{0}W(\eta_{s})\phi\otimes\Omega\|\leq C\beta\int_{0}^{t}\d s\left\{ (\|\eta_{s}\|+\|\eta_{s}\|^{4})(e^{C\|\eta_{s}\|}+1)\right\} \nonumber \\
 & \leq C\beta t\left\{ \log\left[f_{1}( \lambda\kF t)\right]+\log\left[f_{1}( \lambda\kF t)\right]^{4}\right\} \left\{ f_{C}(\lambda \kF t)+1\right\} \label{eq:h_0-bound}
\end{align}
where we used \eqref{eq:eta-bound} and \eqref{eq:exp-eta-bound} in the second inequality.
Together together with $\eqref{eq:Error-estimate}$ and \eqref{eq:Error2-estimate}
inserted in \eqref{eq:almost-bosonic-inequality}, we obtain the desired
result.
\end{proof}
\begin{lem}
\label{lem:Laplacian-estimate}Under the assumptions of \prettyref{thm:Effective-coherent-dynamics},
it holds that for all $t\geq0$ 
\[
\|\Delta_{y}W(\eta_{t})\phi\otimes\Omega\|\leq C(\|\eta_{t}\|+\|\eta_{t}\|^{4})(e^{C\|\eta_{t}\|}+1)
\]
for $C>0$ independent of $\kF$.
\end{lem}

\begin{proof}[Proof of \prettyref{lem:Laplacian-estimate}]
We explicitly calculate the action of the Laplacian on the coupled
coherent state $W(\eta_{s})\psi=e^{B}\psi$ i.e.
\begin{equation}
-\Delta_{y}W(\eta_{s})\psi=-\big(\Delta_{y}W(\eta_{s})\big)\psi-2\beta\nabla_{y}W(\eta_{s})\cdot\nabla_{y}\psi-W(\eta_{s})\Delta_{y}\psi.\label{eq:Laplacian-action}
\end{equation}
In total we expect, that all terms can be bounded by assumption on
the initial condition on $\phi$. We will first focus on the term
$\Delta_{y}W(\eta_{s})$. Recall that it holds
\begin{equation}
W(\eta_{s})^{\ast}\partial_{y_{i}}W(\eta_{s}) =\int_{0}^{1}\d\tau\ e^{-\tau B}\left(\partial_{y_{i}}B_{s}\right)e^{\tau B}
\end{equation}
due to the same calculation as in the proof of \prettyref{lem:derivative-of-Weyl-op}.
With 
\begin{align}
\partial_{y_{i}}B_{s} & =  \partial_{y_{i}}\left\{ c^{\ast}(\eta_{s})-c(\eta_{s})\right\} =c^{\ast}(\partial_{y_{i}}\eta_{s})-c(\partial_{y_{i}}\eta_{s}),\\
\partial_{y_{i}}\eta_{s} & =  \frac{e^{-is\epsilon_{\alpha}(k)}-1}{\epsilon_{\alpha}(k)}\lambda\hat{V}(k)n_{\alpha}(k)ik_{i}e^{iky}=ik_{i}\eta_{s}
\end{align}
it follows analogously to \prettyref{lem:derivative-of-Weyl-op} that
\begin{align}
\partial_{y_{i}}W(\eta_{s}) & =\int_{0}^{1}\d\tau\ e^{(1-\tau)B_{s}}\left(\partial_{y_{i}}B_{s}\right)e^{\tau B_{s}}\nonumber \\
 & =\left(c^{\ast}(ik_{i}\eta_{s})-c(ik_{i}\eta_{s})+i\text{Im}\langle\eta_{s},ik_{i}\eta_{s}\rangle\right)W(\eta_{s})\nonumber \\
 & \hphantom{=}+2i\int_{0}^{1}\d\tau\ W_{1-\tau}(\eta_{s})\text{Im}\langle ik_{i}\eta_{s},\mr{1-\tau}\rangleka W_{\tau}(\eta_{s}).
\end{align}
And repeating the differentiation with 
\begin{align}
\partial_{y_{i}}W_{\tau}(\eta_{s}) & =\tau\left(c^{\ast}(ik_{i}\eta_{s})-c(ik_{i}\eta_{s})+i\tau\text{Im}\langle\eta_{s},ik_{i}\eta_{s}\rangle\right)W_{\tau}(\eta_{s})\nonumber \\
 & \hphantom{=}+2i\int_{0}^{\tau}\d\sigma\ W_{1-\sigma}(\eta_{s})\text{Im}\langle ik_{i}\eta_{s},\mr{\tau-\sigma}\rangleka W_{\sigma}(\eta_{s})
\end{align}
  yields 

\begin{align}
 & \Delta W(\eta_{s})=\sum_{i=1}^{3}\partial_{y_{i}}^{2}W(\eta_{s})\nonumber \\
= & \left(c^{\ast}(k^{2}\eta_{s})-c(k^{2}\eta_{s})+i\text{Im}\langle\eta_{s},k^{2}\eta_{s}\rangle\right)W(\eta_{s})+2i\int_{0}^{1}\d\tau\ W_{1-\tau}(\eta_{s})\text{Im}\langle k^{2}\eta_{s},\mr{1-\tau}\rangleka W_{\tau}(\eta_{s})\nonumber \\
 & +\sum_{i=1}^{3}\left(c^{\ast}(ik_{i}\eta_{s})-c(ik_{i}\eta_{s})+i\text{Im}\langle\eta_{s},ik_{i}\eta_{s}\rangle\right)\partial_{y_{i}}W(\eta_{s})\nonumber \\
 & +2i\sum_{i=1}^{3}\int_{0}^{1}\d\tau\ \text{Im}\langle ik_{i}\eta_{s},\partial_{y_{i}}\left\{ W_{1-\tau}(\eta_{s})\mr{1-\tau}_{\alpha}(k)W_{\tau}(\eta_{s})\right\} \rangleka\\
= & \left(c^{\ast}(k^{2}\eta_{s})-c(k^{2}\eta_{s})+i\text{Im}\langle\eta_{s},k^{2}\eta_{s}\rangle\right)W(\eta_{s})+2i\int_{0}^{1}\d\tau\ \text{Im}\langle k^{2}\eta_{s},\mr{1-\tau}\rangleka W_{\tau}(\eta_{s})\nonumber \\
 & +\sum_{i=1}^{3}\left(c^{\ast}(ik_{i}\eta_{s})-c(ik_{i}\eta_{s})+i\text{Im}\langle\eta_{s},ik_{i}\eta_{s}\rangle\right)\times\nonumber \\
 & \hphantom{-\sum_{i=1}^{3}}\times\left\{ \left(c^{\ast}(ik_{i}\eta_{s})-c(ik_{i}\eta_{s})+i\text{Im}\langle\eta_{s},ik_{i}\eta_{s}\rangle\right)W(\eta_{s})+2i\int_{0}^{1}\d\tau\ \text{Im}\langle ik_{i}\eta_{s},\mr{1-\tau}\rangleka W_{\tau}(\eta_{s})\right\} \nonumber \\
 & +2i\sum_{i=1}^{3}\int_{0}^{1}\d\tau\ \text{Im}\langle ik_{i}\eta_{s},\partial_{y_{i}}\left\{ W_{1-\tau}(\eta_{s})\mr{1-\tau}_{\alpha}(k)W_{\tau}(\eta_{s})\right\} \rangleka\\
\eqqcolon & I_{1}+I_{2}+I_{3}+I_{4}\nonumber 
\end{align}
with 
\begin{align}
I_{1} & \coloneqq\left(c^{\ast}(k^{2}\eta_{s})-c(k^{2}\eta_{s})+i\text{Im}\langle\eta_{s},k^{2}\eta_{s}\rangle\right)W(\eta_{s})\nonumber \\
 & \hphantom{\coloneqq}-\sum_{i=1}^{3}\left(c^{\ast}(ik_{i}\eta_{s})-c(ik_{i}\eta_{s})+i\text{Im}\langle\eta_{s},ik_{i}\eta_{s}\rangle\right)\left(c^{\ast}(ik_{i}\eta_{s})-c(ik_{i}\eta_{s})+i\text{Im}\langle\eta_{s},ik_{i}\eta_{s}\rangle\right)W(\eta_{s}),\\
I_{2} & \coloneqq2i\int_{0}^{1}\d\tau\ W_{1-\tau}(\eta_{s})\text{Im}\langle k^{2}\eta_{s},\mr{1-\tau}\rangleka W_{\tau}(\eta_{s}),\\
I_{3} & \coloneqq2\sum_{i=1}^{3}\left(c^{\ast}(ik_{i}\eta_{s})-c(ik_{i}\eta_{s})+i\text{Im}\langle\eta_{s},ik_{i}\eta_{s}\rangle\right)\int_{0}^{1}\d\tau\ iW_{1-\tau}(\eta_{s})\text{Im}\langle ik_{i}\eta_{s},\mr{1-\tau}\rangleka W_{\tau}(\eta_{s}),\\
I_{4} & \coloneqq2i\sum_{i=1}^{3}\int_{0}^{1}\d\tau\ \text{Im}\langle ik_{i}\eta_{s},\partial_{y_{i}}\left\{ W_{1-\tau}(\eta_{s})\mr{1-\tau}_{\alpha}(k)W_{\tau}(\eta_{s})\right\} \rangleka.
\end{align}

We show that each term can be bounded here by a constant at most of
order $1$.

For $I_{1}$, we can treat all $c^{\ast}(\cdots)$ and $c(\cdots)$
terms with \prettyref{lem:eta-bounds} and \prettyref{lem:eta-bounds2}. Furthermore we use that $\cak\mathcal{N}=(\mathcal{N}+2)\cak$
to estimate
\begin{align}
\|I_{1}\psi\| & \leq C\|\eta_{s}\|\ \|(\mathcal{N}+1)^{1/2}W(\eta_{s})\psi\|+C(\|\eta_{s}\|^{2}+\|\eta_{s}\|^{4})\|W(\eta_{s})\psi\|+C\|\eta_{s}\|^{2}\|(\mathcal{N}+3)W(\eta_{s})\psi\|\nonumber \\
 & \leq C(\|\eta_{s}\|+\|\eta_{s}\|^{2})e^{C\|\eta_{s}\|}\|(\mathcal{N}+5)\psi\|+C(\|\eta_{s}\|^{2}+\|\eta_{s}\|^{4})\label{eq:I1-bound}
\end{align}
where we used \prettyref{prop:Stability-number-op}.

For $I_{2}$, we use a similar approach to \eqref{eq:coherent-state-error-I}
to obtain 
\begin{align}
\|I_{2}\psi\| & \leq C\|\eta_{s}\|MN^{-\frac{2}{3}+\delta}\int_{0}^{1}\d\tau\ (e^{C\|\eta_{s}\|(1-\tau)}-1)\|(\mathcal{N}+3)W_{\tau}(\eta_{s})\psi\|\nonumber \\
 & \leq C\|\eta_{s}\|MN^{-\frac{2}{3}+\delta}\int_{0}^{1}\d\tau\ (e^{C\|\eta_{s}\|}e^{(\tilde{C}-C)\|\eta_{s}\|\tau}-e^{\tilde{C}\|\eta_{s}\|\tau})\|(\mathcal{N}+5)\psi\|\nonumber \\
 & \leq CMN^{-\frac{2}{3}+\delta}e^{C\|\eta_{s}\|}(e^{C\|\eta_{s}\|}-1)\|(\mathcal{N}+5)\psi\|.\label{eq:I2-bound}
\end{align}

For $I_{3}$, we first observe that for $n\in\mathbb{N}$ and $\mu\in [0,1]$
\begin{align}
\|\mathcal{N}^{n}\mr{\mu}_{\beta}(k)\psi\| & \leq\int_{0}^{\mu}\d\tau\ \|\mathcal{N}^{n}e^{-\tau B}\left(\sum_{\linG}\eta_{\beta}(l)\mathcal{E}_{\beta}(l,k)\right)e^{\tau B}\psi\|\nonumber \\
 & \leq\int_{0}^{\mu}\d\tau\ e^{C\|\eta_{s}\|\tau}\sum_{\linG}|\eta_{\beta}(l)|\ \|\mathcal{E}_{\beta}(l,k)(\mathcal{N}+3)^{n}e^{\tau B}\psi\|\nonumber \\
 & \leq C\int_{0}^{\mu}\d\tau\ e^{C\|\eta_{s}\|\tau}\sum_{\linG}|\eta_{\beta}(l)|\ MN^{-\frac{2}{3}+\delta}\|\mathcal{N}(\mathcal{N}+3)^{n}e^{\tau B}\psi\|\nonumber \\
 & \leq C\|\eta_{s}\|^{-1}(e^{C\|\eta_{s}\|\mu}-1)MN^{-\frac{2}{3}+\delta}\sum_{\linG}|\eta_{\beta}(l)|\ \|(\mathcal{N}+5)^{n+1}\psi\|.\label{eq:NR-bound}
\end{align}
We use a similar approach to \eqref{eq:coherent-state-error-II} and
insert the above inequality \eqref{eq:NR-bound} to obtain
\begin{align}
\|I_{3}\psi\| & \leq C\|\eta_{s}\|\sum_{\kinG}\sumak\int_{0}^{1}\d\tau\ \|\left(\|\eta_{s}\|(\mathcal{N}+1)^{1/2}+C\|\eta_{s}\|^{2}\right)W_{1-\tau}(\eta_{s})\mr{1-\tau}_{\alpha}(k)W_{\tau}(\eta_{s})\psi\|\nonumber \\
 & \leq C\|\eta_{s}\|^{2}\sum_{\kinG}\sumak\int_{0}^{1}\d\tau\ \|\mathcal{N}W_{1-\tau}(\eta_{s})\mr{1-\tau}_{\alpha}(k)W_{\tau}(\eta_{s})\psi\|+C\|\eta_{s}\|^{2}\|I_{2}\psi\|\nonumber \\
 & \leq C\|\eta_{s}\|^{2}\sum_{\kinG}\sumak\int_{0}^{1}\d\tau\ e^{C\|\eta_{s}\|(1-\tau)}\|\mathcal{N}\mr{1-\tau}_{\alpha}(k)W_{\tau}(\eta_{s})\psi\|+C\|\eta_{s}\|^{2}\|I_{2}\psi\|\nonumber \\
 & \leq C\|\eta_{s}\|^{2}MN^{-\frac{2}{3}+\delta}\int_{0}^{1}\d\tau\ e^{C\|\eta_{s}\|}(e^{C\|\eta_{s}\|(1-\tau)}-1)\|(\mathcal{N}+5)^{2}\psi\|+C\|\eta_{s}\|^{2}\|I_{2}\psi\|\nonumber \\
 & \leq C\|\eta_{s}\|MN^{-\frac{2}{3}+\delta}e^{C\|\eta_{s}\|}(e^{C\|\eta_{s}\|}-1)\|(\mathcal{N}+5)^{2}\psi\|\nonumber \\
 & \quad+CMN^{-\frac{2}{3}+\delta}\|\eta_{s}\|^{2}e^{C\|\eta_{s}\|}(e^{C\|\eta_{s}\|}-1)\|(\mathcal{N}+5)\psi\|\nonumber \\
 & \leq CMN^{-\frac{2}{3}+\delta}(\|\eta_{s}\|+\|\eta_{s}\|^{2})e^{C\|\eta_{s}\|}(e^{C\|\eta_{s}\|}-1)\|(\mathcal{N}+5)^{2}\psi\|.\label{eq:I3-bound}
\end{align}

For $I_{4}$, we first calculate 
\begin{align}
 & \partial_{y_{i}}\left\{ W_{1-\tau}(\eta_{s})\mr{1-\tau}_{\alpha}(k)W_{\tau}(\eta_{s})\right\} \nonumber \\
 & =\int_{\tau}^{1}\d\sigma\ \partial_{y_{i}}e^{(1-\sigma)B}\left(\sum_{\linG}\eta_{\alpha}(l)\mathcal{E}_{\alpha}(l,k)\right)e^{\sigma B}+\int_{\tau}^{1}\d\sigma\ e^{(1-\sigma)B}\left(\sum_{\linG}ik_{i}\eta_{\alpha}(l)\mathcal{E}_{\alpha}(l,k)\right)e^{\sigma B}\nonumber \\
 & \hphantom{=}+\int_{\tau}^{1}\d\sigma\ e^{(1-\sigma)B}\left(\sum_{\linG}\eta_{\alpha}(l)\mathcal{E}_{\alpha}(l,k)\right)\partial_{y_{i}}e^{\sigma B}\\
 & \eqqcolon I_{4,1}+I_{4,2}+I_{4,3}+I_{4,4}+I_{4,5}\nonumber 
\end{align}
with 
\begin{align}
I_{4,1} & =\int_{\tau}^{1}\d\sigma\ (1-\sigma)\left(c^{\ast}(ik_{i}\eta_{s})-c(ik_{i}\eta_{s})+i(1-\sigma)\text{Im}\langle\eta_{s},ik_{i}\eta_{s}\rangle\right)\nonumber \\
 & \hphantom{=\int_{\tau}^{1}\d\sigma\ (1-\sigma)}\times e^{(1-\sigma)B}\left(\sum_{\linG}\eta_{\alpha}(l)\mathcal{E}_{\alpha}(l,k)\right)e^{\sigma B},\\
I_{4,2} & =2i\int_{\tau}^{1}\d\sigma\int_{0}^{1-\sigma}\d a\ e^{(1-a)B}\text{Im}\langle ik_{i}\eta_{s},\mr{1-\sigma-a}\rangleka e^{aB}\left(\sum_{\linG}\eta_{\alpha}(l)\mathcal{E}_{\alpha}(l,k)\right)e^{\sigma B},\\
I_{4,3} & =\int_{\tau}^{1}\d\sigma\ e^{(1-\sigma)B}\left(\sum_{\linG}ik_{i}\eta_{\alpha}(l)\mathcal{E}_{\alpha}(l,k)\right)e^{\sigma B},\\
I_{4,4} & =\int_{\tau}^{1}\d\sigma\ e^{(1-\sigma)B}\left(\sum_{\linG}\eta_{\alpha}(l)\mathcal{E}_{\alpha}(l,k)\right)\sigma\left(c^{\ast}(ik_{i}\eta_{s})-c(ik_{i}\eta_{s})+i\sigma\text{Im}\langle\eta_{s},ik_{i}\eta_{s}\rangle\right)e^{\sigma B},\\
I_{4,5} & =2i\int_{\tau}^{1}\d\sigma\ e^{(1-\sigma)B}\left(\sum_{\linG}\eta_{\alpha}(l)\mathcal{E}_{\alpha}(l,k)\right)\int_{0}^{\sigma}\d a\ e^{(1-a)B}\text{Im}\langle ik_{i}\eta_{s},\mr{\sigma-a}\rangleka e^{aB}.
\end{align}

We approach each term similarly to \eqref{eq:coherent-state-error-I}.

For $I_{4,1}$, it holds 
\begin{align}
 & \|I_{4,1}\psi\|\nonumber \\
 & \leq C\sum_{\linG}|\eta_{\alpha}(l)|\int_{\tau}^{1}\d\sigma\ (1-\sigma)\left(\|\eta_{s}\|\ \|(\mathcal{N}+1)e^{(1-\sigma)B}\mathcal{E}_{\alpha}(l,k)e^{\sigma B}\psi\|+C\|\eta_{s}\|^{2}\|\mathcal{E}_{\alpha}(l,k)e^{\sigma B}\psi\|\right)\nonumber \\
 & \leq C\|\eta_{s}\|\sum_{\linG}|\eta_{\alpha}(l)|\int_{\tau}^{1}\d\sigma\ (1-\sigma)\left(e^{C\|\eta_{s}\|(1-\sigma)}\|(\mathcal{N}+3)\mathcal{E}_{\alpha}(l,k)e^{\sigma B}\psi\|+C\|\eta_{s}\|\,\|\mathcal{E}_{\alpha}(l,k)e^{\sigma B}\psi\|\right)\nonumber \\
 & \leq C\|\eta_{s}\|MN^{-\frac{2}{3}+\delta}\sum_{\linG}|\eta_{\alpha}(l)|\int_{\tau}^{1}\d\sigma\ \left((1-\sigma)e^{C\|\eta_{s}\|(1-\sigma)}+C\|\eta_{s}\|\right)\|(\mathcal{N}+3)^{2}e^{\sigma B}\psi\|\nonumber \\
 & \leq C\|\eta_{s}\|MN^{-\frac{2}{3}+\delta}\sum_{\linG}|\eta_{\alpha}(l)|\int_{\tau}^{1}\d\sigma\ \left((1-\sigma)e^{C\|\eta_{s}\|(1-\sigma)}+C\|\eta_{s}\|\right)e^{\tilde{C}\|\eta_{s}\|\sigma}\|(\mathcal{N}+5)^{2}\psi\|\\
 & \leq CMN^{-\frac{2}{3}+\delta}\sum_{\linG}|\eta_{\alpha}(l)|\left((1-\tau)e^{C\|\eta_{s}\|(1-\tau)}+C\|\eta_{s}\|e^{C\|\eta_{s}\|}\right)\|(\mathcal{N}+5)^{2}\psi\|
\end{align}
where we used \prettyref{lem:Pair-operator-bounds}, \prettyref{lem:eta-bounds} and \prettyref{lem:eta-bounds2}
in the first inequality, \prettyref{lem:CAR-error} in the third inequality
and \prettyref{prop:Stability-number-op} in the second and forth
inequality. Therefore using $\int_{0}^{1}(1-\tau)e^{y(1-\tau)}\d\tau=y^{-2}\left((y-1)e^{y}+1\right)$
yields
\begin{align}
 & \int_{0}^{1}\d\tau\ \|\eta_{s}\|\sqrt{\sum_{\kinG}\sum_{\alpha\in\mathcal{I}_{k}\cap\mathcal{I}_{l}}\|I_{4,1}\psi\|^{2}}\\
 & \leq CMN^{-\frac{2}{3}+\delta}\left((\|\eta_{s}\|^{3}+\|\eta_{s}\|-1)e^{C\|\eta_{s}\|}+1\right)\|(\mathcal{N}+5)^{2}\psi\|
\end{align}
where we used $\sumal\bigg(\sum_{\linG}|\eta_{\alpha}(l)|\bigg)^{2}\leq\sum_{l',\linG}\|\eta(l)\|_{l^{2}}\|\eta(l')\|_{l^{2}}\leq C\|\eta\|^{2}$
by \prettyref{lem:eta-bounds} and \prettyref{lem:eta-bounds2}.

For $I_{4,2}$, it holds
\begin{align}
 & \|I_{4,2}\psi\|\nonumber \\
 & \leq C\|\eta_{s}\|MN^{-\frac{2}{3}+\delta}\sum_{\linG}|\eta_{\alpha}(l)|\int_{\tau}^{1}\d\sigma\int_{0}^{1-\sigma}\d a\ \|(\mathcal{N}+3)e^{aB}\mathcal{E}_{\alpha}(l,k)e^{\sigma B}\psi\|\nonumber \\
 & \leq C\|\eta_{s}\|\left(MN^{-\frac{2}{3}+\delta}\right)^{2}\sum_{\linG}|\eta_{\alpha}(l)|\int_{\tau}^{1}\d\sigma\int_{0}^{1-\sigma}\d a\ e^{C\|\eta_{s}\|a}\|(\mathcal{N}+5)^{2}e^{\sigma B}\psi\|\nonumber \\
 & \leq C\left(MN^{-\frac{2}{3}+\delta}\right)^{2}\sum_{\linG}|\eta_{\alpha}(l)|\int_{\tau}^{1}\d\sigma(e^{C\|\eta_{s}\|(1-\sigma)}-1)e^{\tilde{C}\|\eta_{s}\|\sigma}\|(\mathcal{N}+8)^{2}\psi\|\nonumber \\
 & \leq C\|\eta_{s}\|^{-1}\left(MN^{-\frac{2}{3}+\delta}\right)^{2}\sum_{\linG}|\eta_{\alpha}(l)|e^{C\|\eta_{s}\|(1-\tau)}\|(\mathcal{N}+8)^{2}\psi\|
\end{align}
where we used Cauchy-Schwarz with \eqref{eq:coherent-state-error-I}, \prettyref{lem:eta-bounds} and \prettyref{lem:eta-bounds2} in the first inequality, \prettyref{prop:Stability-number-op}
and \prettyref{lem:CAR-error} in the second inequality. Therefore
it follows
\begin{equation}
 \int_{0}^{1}\d\tau\ \|\eta_{s}\|\sqrt{\sum_{\kinG}\sum_{\alpha\in\mathcal{I}_{k}\cap\mathcal{I}_{l}}\|I_{4,2}\psi\|^{2}}\nonumber \leq CM^{2}N^{-\frac{4}{3}+2\delta}(e^{C\|\eta_{s}\|}-1)\|(\mathcal{N}+8)^{2}\psi\|.
\end{equation}

The term $I_{4,3}$ is estimated similarly to \eqref{eq:coherent-state-error-I}
by
\begin{equation}
\|I_{4,3}\psi\|\leq C\|\eta_{s}\|^{-1}MN^{-\frac{2}{3}+\delta}\sum_{\linG}|\eta_{\alpha}(l)|(e^{C\|\eta_{s}\|}-e^{C\|\eta_{s}\|\tau})\|(\mathcal{N}+3)\psi\|
\end{equation}
and therefore
\begin{align}
 & \int_{0}^{1}\d\tau\ \|\eta_{s}\|\sqrt{\sum_{\kinG}\sum_{\alpha\in\mathcal{I}_{k}\cap\mathcal{I}_{l}}\|I_{4,3}\psi\|^{2}}\leq CMN^{-\frac{2}{3}+\delta}\|(\mathcal{N}+3)\psi\|.
\end{align}

Similarly for $I_{4,4}$, we estimate 
\begin{align}
 & \|I_{4,4}\psi\|\nonumber \\
 & \leq CMN^{-\frac{2}{3}+\delta}\sum_{\linG}|\eta_{\alpha}(l)|\int_{\tau}^{1}\d\sigma\ \sigma\|\mathcal{N}\left(c^{\ast}(ik_{i}\eta_{s})-c(ik_{i}\eta_{s})+i\sigma\text{Im}\langle\eta_{s},ik_{i}\eta_{s}\rangle\right)e^{\sigma B}\psi\|\nonumber \\
 & \leq C(\|\eta_{s}\|+\|\eta_{s}\|^{2})MN^{-\frac{2}{3}+\delta}\sum_{\linG}|\eta_{\alpha}(l)|\int_{\tau}^{1}\d\sigma\ \sigma\|(\mathcal{N}+1)^{2}e^{\sigma B}\psi\|\nonumber \\
 & \leq C(\|\eta_{s}\|^{-1}+1)MN^{-\frac{2}{3}+\delta}\sum_{\linG}|\eta_{\alpha}(l)|\ \left(e^{C\|\eta_{s}\|}+(1-C\|\eta_{s}\|\tau)e^{C\|\eta_{s}\|\tau}\right)\|(\mathcal{N}+3)^{2}\psi\|\nonumber \\
 & \leq C(\|\eta_{s}\|^{-1}+1)MN^{-\frac{2}{3}+\delta}\sum_{\linG}|\eta_{\alpha}(l)|\ e^{C\|\eta_{s}\|}\|(\mathcal{N}+3)^{2}\psi\|
\end{align}
using $\int_{\tau}^{1}\d\sigma\ \sigma e^{y\sigma}=y^{-2}\left((y-1)e^{y}-(\tau y-1)e^{y\tau}\right)$
and therefore
\begin{align}
 & \int_{0}^{1}\d\tau\ \|\eta_{s}\|\sqrt{\sum_{\kinG}\sum_{\alpha\in\mathcal{I}_{k}\cap\mathcal{I}_{l}}\|I_{4,4}\psi\|^{2}}\nonumber \\
 & \leq CMN^{-\frac{2}{3}+\delta}(\|\eta_{s}\|+\|\eta_{s}\|^{2})e^{C\|\eta_{s}\|}\|(\mathcal{N}+1)^{2}\psi\|.
\end{align}

For $I_{4,5}$, we estimate
\begin{align}
\|I_{4,5}\psi\| & =\|2i\int_{\tau}^{1}\d\sigma\ e^{(1-\sigma)B}\left(\sum_{\linG}\eta_{\alpha}(l)\mathcal{E}_{\alpha}(l,k)\right)\int_{0}^{\sigma}\d a\ e^{(1-a)B}\text{Im}\langle ik_{i}\eta_{s},\mr{\sigma-a}\rangleka e^{aB}\psi\|\nonumber \\
 & \leq CMN^{-\frac{2}{3}+\delta}\sum_{\linG}|\eta_{\alpha}(l)|\int_{\tau}^{1}\d\sigma\int_{0}^{\sigma}\d a\ \|\mathcal{N}e^{(1-a)B}\text{Im}\langle ik_{i}\eta_{s},\mr{\sigma-a}\rangleka e^{aB}\psi\|\nonumber \\
 & \leq C\|\eta_{s}\|MN^{-\frac{2}{3}+\delta}\sum_{\linG}|\eta_{\alpha}(l)|\int_{\tau}^{1}\d\sigma\int_{0}^{\sigma}\d a\ e^{C\|\eta_{s}\|(1-a)}\bigg(\sumak\|(\mathcal{N}+3)\mr{\sigma-a}_{\alpha}(k)e^{aB}\psi\|^{2}\bigg)^{1/2}\nonumber \\
 & \leq C\|\eta_{s}\|M^{2}N^{-\frac{4}{3}+2\delta}\sum_{\linG}|\eta_{\alpha}(l)|\int_{\tau}^{1}\d\sigma\int_{0}^{\sigma}\d a\ e^{C\|\eta_{s}\|(1-a)}(e^{C\|\eta_{s}\|(\sigma-a)}-1)\|(\mathcal{N}+8)^{2}e^{aB}\psi\|\nonumber \\
 & \leq C\|\eta_{s}\|M^{2}N^{-\frac{4}{3}+2\delta}\sum_{\linG}|\eta_{\alpha}(l)|\int_{\tau}^{1}\d\sigma\int_{0}^{\sigma}\d a\ e^{C\|\eta_{s}\|}e^{C'\|\eta_{s}\|\sigma}e^{-C''\|\eta_{s}\|a}\|(\mathcal{N}+10)^{2}\psi\|\nonumber \\
 & \leq C\|\eta_{s}\|^{-1}M^{2}N^{-\frac{4}{3}+2\delta}\sum_{\linG}|\eta_{\alpha}(l)|\ (e^{C\|\eta_{s}\|}-e^{C\|\eta_{s}\|\tau})\|(\mathcal{N}+10)^{2}\psi\|
\end{align}
where we used \eqref{eq:NR-bound} in the third inequality. Therefore
we obtain
\begin{equation}
\int_{0}^{1}\d\tau\ \|\eta_{s}\|\sqrt{\sum_{\kinG}\sum_{\alpha\in\mathcal{I}_{k}\cap\mathcal{I}_{l}}\|I_{4,5}\psi\|^{2}}\leq CM^{2}N^{-\frac{4}{3}+2\delta}\|(\mathcal{N}+10)^{2}\psi\|.
\end{equation}

Taking the five bounds together we finally obtain
\begin{align}
\|I_{4}\psi\| & =\|2i\sum_{i=1}^{3}\int_{0}^{1}\d\tau\ \text{Im}\langle ik_{i}\eta_{s},\partial_{y_{i}}\left\{ W_{1-\tau}(\eta_{s})\mr{1-\tau}_{\alpha}(k)W_{\tau}(\eta_{s})\right\} \rangleka\psi\|\nonumber \\
 & \leq C\int_{0}^{1}\d\tau\sum_{n=1}^{5}\|\eta_{s}\|\bigg(\sum_{\kinG}\sum_{\alpha\in\mathcal{I}_{k}\cap\mathcal{I}_{l}}\|I_{4,n}\psi\|^{2}\bigg)^{1/2}\nonumber \\
 & \leq CMN^{-\frac{2}{3}+\delta}\left((\|\eta_{s}\|^{3}+\|\eta_{s}\|-1)e^{C\|\eta_{s}\|}+1\right)\|(\mathcal{N}+8)^{2}\psi\|\label{eq:I4-bound}.
\end{align}

Combining \eqref{eq:I1-bound}, \eqref{eq:I2-bound}, \eqref{eq:I3-bound}
and \eqref{eq:I4-bound} with $\psi\equiv\phi\otimes\Omega$ yields
the desired final result of
\begin{equation}
\|\big(\Delta_{y}W(\eta_{s})\big)\psi\|\leq C(\|\eta_{s}\|+\|\eta_{s}\|^{4})(e^{C\|\eta_{s}\|}+1).
\end{equation}

Similarly to \eqref{eq:I1-bound} and \eqref{eq:I2-bound}, we further
estimate the second term of \eqref{eq:Laplacian-action}
\begin{align}
 & \|\nabla_{y}W(\eta_{s})\cdot\nabla_{y}\psi\|\nonumber \\
& =  \|\sum_{i=1}^{3}\Big(\left(c^{\ast}(ik_{i}\eta_{s})-c(ik_{i}\eta_{s})+i\text{Im}\langle\eta_{s},ik_{i}\eta_{s}\rangle\right)W(\eta_{s})\nonumber \\
 & \hphantom{\|\sum_{i=1}^{3}\Big(}  +2i\int_{0}^{1}\d\tau\ \text{Im}\langle ik_{i}\eta_{s},\mr{1-\tau}\rangleka W_{\tau}(\eta_{s})\Big)\partial_{y_{i}}\psi\|\nonumber \\
&\leq  C\left((\|\eta_{s}\|+\|\eta_{s}\|^{2})e^{C\|\eta_{s}\|}+(\|\eta_{s}\|^{2}+\|\eta_{s}\|^{4})\right)\sum_{i=1}^{3}\|\partial_{y_{i}}\phi\|\nonumber \\
& \hphantom{\|\sum_{i=1}^{3}\Big(}  +CMN^{-\frac{2}{3}+\delta}e^{C\|\eta_{s}\|}(e^{C\|\eta_{s}\|}-1)\sum_{i=1}^{3}\|\partial_{y_{i}}\phi\| \\
&\leq C(\|\eta_{s}\|+\|\eta_{s}\|^{4})(e^{C\|\eta_{s}\|}+1)\sum_{i=1}^{3}\|\partial_{y_{i}}\phi\|.
\end{align}
Therefore in total for all $t\in\R$
\begin{align}
 & \|\Delta_{y}W(\eta_{t})\phi\otimes\Omega\|\\
 & \leq C\left\{ (\|\eta_{s}\|+\|\eta_{s}\|^{4})(e^{C\|\eta_{s}\|}+1)+\sum_{i=1}^{3}\big((\|\eta_{s}\|+\|\eta_{s}\|^{4})(e^{C\|\eta_{s}\|}+1)\|\partial_{y_{i}}\phi\|+\|\partial_{y_{i}}^{2}\phi\|\big)\right\} .
\end{align}
\end{proof}

\paragraph*{Proofs of properties of $\eta$}
\begin{proof}[Proof of \prettyref{lem:eta-bounds}]
  Recall that by definition it holds
  \begin{equation}
  \|\eta_{s}\|^{2}=\sum_{\kinG}\sumak|\big(\eta_{s}\big)_{\alpha}(k)|^{2}=\lambda^{2}\sum_{\kinG}|\hat{V}(k)|^{2}\underbrace{\sumak\bigg|n_{\alpha}(k)\frac{\sin\big(\epsilon_{\alpha}(k)s/2\big)}{\epsilon_{\alpha}(k)/2}\bigg|^{2}}_{\eqqcolon S_{k}}\label{eq:Sk}
  \end{equation}
  with $\epsilon_{\alpha}(k)=2 k_{F}|k\cdot\omega_{\alpha}|$.
  We will first approximate $S_{k}$ and then give a upper bound. 
  
  Firstly, we approximate the $\alpha$-sum $S_{k}$ by an integral
  by identifying 
  \begin{equation}
  n_{\alpha}(k)^{2}=\kF^{2}|k|\sigma(p_{\alpha})u_{\alpha}(k)^{2}\left(1+\mathcal{O}(M^{\frac{1}{2}}N^{-\frac{1}{3}+\delta})\right)\label{eq:recall-n_alpha(k)}
  \end{equation}
  with $\cos\theta_{\alpha}\coloneqq|\hat{k}\cdot\hat{\omega}_{\alpha}|\equiv u_{\alpha}(k)^{2}$
  analogously to \prettyref{lem:Approx-n_alpha(k)^2-wo-sum}. Thus,
  we calculate the half-sphere integral as approximation for $S_{k}$
  \begin{align}
  S_{k} & =\frac{1}{|k|}\int_{0}^{2\pi}\d\varphi\int_{0}^{\pi/2}\d\theta\frac{\sin(\kF|k|s\cos\theta)^{2}}{\cos\theta}\sin\theta+\mathcal{E}\nonumber \\
   & =\frac{2\pi}{|k|}\int_{0}^{1}\d u\frac{\sin(\kF|k|su)^{2}}{u}+\mathcal{E}\nonumber \\
   & =\frac{\pi}{|k|}\left\{ \log(2\kF|k|s)-\text{Ci}(2\kF|k|s)+\gamma\right\} +\mathcal{E}
  \end{align}
  with total error $\mathcal{E}=\mathcal{E}_{1}+\mathcal{E}_{2}+\mathcal{E}_{3}$
  consisting three terms: the error $\mathcal{E}_{1}$ from \prettyref{eq:recall-n_alpha(k)},
  the error $\mathcal{E}_{2}$ from approximating the discrete variables
  $\theta_{\alpha}$ by continuous $\theta$ given by
  \begin{align}
  \mathcal{E}_{2} & =\left|\int_{p_{\alpha}}\d\sigma\ \frac{\sin(\kF|k|s\cos\theta)^{2}}{\cos\theta}-\sigma(p_{\alpha})\frac{\sin(\kF|k|s\cos\theta_{\alpha})^{2}}{\cos\theta_{\alpha}}\right|\nonumber \\
   & \leq\sup_{\hat{\omega}(\theta,\varphi)\in p_{\alpha}}\left|\frac{\d}{\d\theta}\frac{\sin(\kF|k|s\cos\theta)^{2}}{\cos\theta}\right|\sup_{(\theta,\varphi)\in p_{\alpha}}\left|\theta-\theta_{\alpha}\right|\sigma(p_{\alpha})\nonumber \\
   & \leq C(\kF|k|s)^{2}M^{-\frac{3}{2}}
  \end{align}
  and the error $\mathcal{E}_{3}$ from the patch construction which
  is given by
  \begin{equation}
  \mathcal{E}_{3}=\sup_{\hat{\omega}(\theta,\varphi)\in p_{\alpha}}\left|\frac{\sin(\kF|k|s\cos\theta)^{2}}{\cos\theta}\right|\left|\sigma\Big(\bigcup_{\alpha\in\mathcal{I}_{k}}p_{\alpha}\Big)-\sigma(S)\right|\leq C\kF|k|s\left(N^{-\delta}+M^{\frac{1}{2}}N^{-\frac{1}{3}}\right).
  \end{equation}
  Thus with the triangle inequality 
  \begin{align}
   & \left|S_{k}-\frac{\pi}{|k|}\left\{ \log(2\kF|k|s)-\text{Ci}(2\kF|k|s)+\gamma\right\} \right|\nonumber \\
   & \leq\frac{C}{|k|}\left\{ M^{\frac{1}{2}}N^{-\frac{1}{3}+\delta}+(\kF|k|s)^{2}M^{-\frac{3}{2}}+\kF|k|s\left(N^{-\delta}+M^{\frac{1}{2}}N^{-\frac{1}{3}}\right)\right\} .
  \end{align}

  We will now give an estimate for $S_{k}$ by approximating for any $m \in \R$
  \[
  \frac{\sin(m x)^{2}}{x}=\frac{1}{2}\frac{1-\cos(2m x)}{x}\leq\begin{cases}
  m^{2}x & \text{if } 2m x\in[0,\pi/2],\\
  \frac{1}{x} & \text{if } 2m x>\pi/2
  \end{cases}
  \]
  where we used $1-x^{2}/2\leq\cos(x)$ for all $x\in[0,\pi/2]$. Therefore
  \begin{align}
   & \frac{2\pi}{|k|}\int_{0}^{1}\d u\frac{\sin(\kF|k|su)^{2}}{u}\nonumber \\
   & \leq\frac{2\pi}{|k|}\int_{0}^{1}\d u\ \left\{ (\kF|k|s)^{2}u\ \chi(2\kF|k|su\leq\pi/2)+\frac{1}{u}\chi(2\kF|k|su>\pi/2)\right\} \nonumber \\
   & \leq\frac{2\pi}{|k|}\int_{0}^{1}\d u\ \left\{ (\kF|k|s)^{2}u\ \chi\Big(u\leq\frac{\pi}{4\kF|k|s}\Big)+\frac{1}{u}\ \chi\Big(u>\frac{\pi}{4\kF|k|s}\Big)\right\} \nonumber \\
   & =\frac{2\pi}{|k|}\left\{ \frac{1}{2}(\kF|k|s)^{2}\Big(\frac{\pi}{4\kF|k|s}\Big)^{2}-\log\Big(\frac{\pi}{4\kF|k|s}\Big)\right\} \ \chi\Big(1>\frac{\pi}{4\kF|k|s}\Big)\nonumber \\
   & \qquad+\frac{2\pi}{|k|}\frac{1}{2}(\kF|k|s)^{2}\ \chi\Big(1\leq\frac{\pi}{4\kF|k|s}\Big)\nonumber \\
   & =\frac{2\pi}{|k|}\left\{ \frac{\pi^{2}}{32}-\log(\pi)+\log(4\kF|k|s)\right\} \ \chi\Big(\kF|k|s>\frac{\pi}{4}\Big)\nonumber \\
   & \qquad+\frac{2\pi}{|k|}\frac{1}{2}(\kF|k|s)^{2}\ \chi\Big(\kF|k|s\leq\frac{\pi}{4}\Big)\nonumber \\
   & \leq\frac{2\pi}{|k|}\min\left\{ \frac{1}{2}(\kF|k|s)^{2},\log(4\kF|k|s+1)\right\} .
  \end{align}
  Note that $\log(4\kF|k|s+1)\leq\log(4\kF s+1)+\log|k|$
  using $\frac{1}{2}x^{2}\leq\log(4x+1)$ for all $x\leq2$ and $|k|\geq1$
  for all $k\in\mathbb{Z}_{\ast}^{3}$. Thus, it holds
  \begin{align}
  \|\eta_{s}\|^{2} & \leq2\pi \lambda^2 \sum_{k\in\mathbb{Z}_{\ast}^{3}}|\hat{V}(k)|^{2}\min\left\{ \frac{(\kF s)^{2}}{2}|k|,\frac{\log(4\kF|k|s+1)}{|k|}\right\} \nonumber \\
   & \leq2\pi \lambda^2 \min\left\{ \|\hat{V}(\cdot)^{1/2}\|_{2}^{2}(\kF s)^{2}/2,\|\hat{V}\|_{2}^{2}\left(\log(4\kF s+1)+1\right)\right\} .
  \end{align}
  
  Thus by using $\sqrt{x+y}\leq\sqrt{x}+\sqrt{y}$ and $\sqrt{\log(x+1)}\leq\log(x+2)$
  we obtain the desired
  \begin{align}
  \|\eta_{s}\| & \leq \lambda \min\left\{ b\kF s,\log(4\kF s+2)^{a}+a\right\} ,\\
  e^{c_{0}\|\eta_{s}\|} & \leq  \min\left\{ e^{b c_{0} \lambda \kF s},e^{a c_{0} \lambda}(4\kF s+2)^{a c_{0} \lambda}\right\} 
  \end{align}
  with $a=\sqrt{2\pi}\|\hat{V}\|_{2},b=\sqrt{\pi}\|\hat{V}(\cdot)^{1/2}\|_{2}.$
  \end{proof}

  \begin{proof}[Proof of \prettyref{lem:eta-bounds2}]
    For the inequalities, we observe that for $n\in\mathbb{N}_{0}$ it
    holds
    \begin{align}
    \sum_{\kinG}\||k|^{n}\eta_{s}(k)\|_{l^{2}} & \equiv\sum_{\kinG}|k|^{n}\Big(\sumak|(\eta_{s})_{\alpha}(k)|^{2}\Big)^{1/2}=\lambda\sum_{\kinG}|k|^{n}|\hat{V}(k)|\left(S_{k}\right)^{1/2}
    \end{align}
    with $S_k$ from \eqref{eq:Sk}
    Since by assumption $\|(\cdot)^{n}\hat{V}\|_{1}$ is bounded for each
    $n\in\mathbb{N}$, the first statement follows. The second statement
    simply follows from the same calculation and recalling that $\langle\eta_{s},|k|^{n}\eta_{s}\rangle\equiv\sum_{\kinG}\sumak|k|^{n}|\eta_{\alpha}(k)|^{2}$.
    The third statement follows from \prettyref{lem:Pair-operator-bounds}
    and 
    \begin{equation}
    \|c^{\ast}(|k|^{n}\eta_{s})\psi\|\leq\sum_{\kinG}\|\sumak|k|^{n}\eta_{\alpha}(k)\ccak\psi\|\leq\sum_{\kinG}\||k|^{n}\eta_{s}(k)\|_{l^{2}}\|(\mathcal{N}+1)^{1/2}\psi\|.
    \end{equation}
    \end{proof}

\section{Proof of \prettyref{cor:Corollary}}
\begin{proof}[Proof of \prettyref{cor:Corollary}]
We observe that by the inverse triangle inequality it holds 
\begin{subequations}
\begin{align}
 & \|R^{\ast}e^{-i\mathbb{H}t}R\psi-e^{-i\widetilde{\mathbb{H}^{\text{eff}}}t}\psi\|\nonumber \\
 & \geq\|e^{iP(t)}W(\eta_{t})\psi-e^{-i\widetilde{\mathbb{H}^{\text{eff}}}t}\psi\|\label{eq:inverse-triangle-a}\\
 & \hphantom{\geq}-\|e^{-i\mathbb{H}^{\text{eff}}t}\psi-e^{iP(t)}W(\eta_{t})\psi\|-\|R^{\ast}e^{-i\mathbb{H}t}R\psi-e^{-i\mathbb{H}^{\text{eff}}t}\psi\|.\label{eq:inverse-triangle-b}
\end{align}
\end{subequations}
The second line \eqref{eq:inverse-triangle-b} is to be bounded from
below by bounds of order $o(1)$ from \prettyref{thm:Main-thm} and
\prettyref{thm:Effective-coherent-dynamics}. Therefore it is sufficient
to show that the first line \eqref{eq:inverse-triangle-a} is large.

The first line \eqref{eq:inverse-triangle-a} can be explicitly estimated
by the same approach as in the proof of \prettyref{thm:Effective-coherent-dynamics}.
That is use Duhamel's formula, commute all $\cak$-operators with
$W(\eta_{s})$ and collect the error terms via \prettyref{lem:Approximate-shift}:
\begin{align}
 & \bigg(e^{i\widetilde{\mathbb{H}^{\text{eff}}}t}e^{iP(t)}W(\eta_{t})-1\bigg)\psi\nonumber \\
 & =i\int_{0}^{t}\d s\ e^{i\widetilde{\mathbb{H}^{\text{eff}}}s}e^{iP(s)}\left( \big(\widetilde{\mathbb{H}^{\text{eff}}}-E_{N}^{\text{pw}}+2\text{Im}(\dot{\nu}_{s})-\text{Im}\langle\eta_{s},\dot{\eta}_{s}\rangle\big)W(\eta_{s})\psi-i\partial_{s}W(\eta_{s})\psi\right) \\
 & =i\int_{0}^{t}\d s\ e^{i\widetilde{\mathbb{H}^{\text{eff}}}s}e^{iP(s)}\bigg(\big(h_{0}-c^{\ast}(h_{y})-c(h_{y})\big)W(\eta_{s})\psi+\text{Error}\bigg)\\
 & =i\int_{0}^{t}\d s\ e^{i\widetilde{\mathbb{H}^{\text{eff}}}s}e^{iP(s)}\bigg(-W(\eta_{s})\left(c^{\ast}(h_{y})+\langle\eta_{s},h_{y}\rangle+\langle h_{y},\eta_{s}\rangle\right)\psi+h_{0}W(\eta_{s})\psi+\widetilde{\text{Error}}\bigg)
\end{align}
with $\int_{0}^{t}\d s\ \text{Error}=\text{Error}_{1}+\text{Error}_{2}$
from \eqref{eq:almost-bosonic-inequality} and where we also commuted
$c^{\ast}(h_{y})$ and $c(h_{y})$ with $W(\eta_{s})$ such that the
new error term is of the form 
\begin{align}
\widetilde{\text{Error}} & \coloneqq\sum_{\kinG}\sumak\left(\epsilon_{\alpha}(k)\ccak-e^{is\epsilon_{\alpha}(k)}\right)\overline{\phiak}W(\eta_{s})\mr{1}_{\alpha}(k)\psi\nonumber \\
 & \hphantom{\coloneqq}-\sum_{\kinG}\sumak W(\eta_{s})\mr{1}_{\alpha}(k)^{\ast}\phiak\psi\nonumber \\
 & \hphantom{\coloneqq}-2i\int_{0}^{1}\d\tau\ W_{1-\tau}(\eta_{s})\text{Im}\langle\dot{\eta}_{s},\mr{1-\tau}\rangleka W_{\tau}(\eta_{s})\psi.
\end{align}
Note that $\widetilde{\text{Error}}$ is $s$-dependent even though
we did not include the dependence explicitly in the notation.

We employ the Cauchy-Schwarz inequality, insert our previous finding
with the triangle inequality and use that $e^{-i\widetilde{\mathbb{H}^{\text{eff}}}s}$
is unitary to estimate
\begin{align}
 & \|e^{iP(t)}W(\eta_{t})\psi-e^{-i\widetilde{\mathbb{H}^{\text{eff}}}t}\psi\|\nonumber \\
 & =\|\big(e^{i\widetilde{\mathbb{H}^{\text{eff}}}t}e^{iP(t)}W(\eta_{t})-1\big)\psi\|\geq\left|\langle\psi,\big(1-e^{i\widetilde{\mathbb{H}^{\text{eff}}}t}e^{iP(t)}W(\eta_{t})\big)\psi\rangle\right|\nonumber \\
 & =\left|\int_{0}^{t}\d s\langle e^{-iP(s)}e^{-i\widetilde{\mathbb{H}^{\text{eff}}}s}\psi,\big(\widetilde{\mathbb{H}^{\text{eff}}}-E_{N}^{\text{pw}}+2\text{Im}(\dot{\nu}_{s})-\text{Im}\langle\eta_{s},\dot{\eta}_{s}\rangle\big)W(\eta_{s})\psi-i\partial_{s}W(\eta_{s})\psi\rangle\right|\nonumber \\
 & \geq\left|\int_{0}^{t}\d s\langle\psi,e^{i\widetilde{\mathbb{H}^{\text{eff}}}s}e^{iP(s)}W(\eta_{s})\left(c^{\ast}(h_{y})+2\text{Re}\langle h_{y},\eta_{s}\rangle\right)\psi\rangle\right|\nonumber \\
 & \hspace{4cm}-\int_{0}^{t}\d s\left|\langle e^{-iP(s)}e^{-i\widetilde{\mathbb{H}^{\text{eff}}}s}\psi,\widetilde{\text{Error}}+h_{0}W(\eta_{s})\psi\rangle\right|\\
 & \geq\left|\int_{0}^{t}\d s\langle\psi,\big(1+e^{i\widetilde{\mathbb{H}^{\text{eff}}}s}e^{iP(s)}W(\eta_{s})-1\big)\left(c^{\ast}(h_{y})+2\text{Re}\langle h_{y},\eta_{s}\rangle\right)\psi\rangle\right|\nonumber \\
 & \hspace{4cm}-\int_{0}^{t}\d s\left\{ \|\widetilde{\text{Error}}\|+\|h_{0}W(\eta_{s})\psi\|\right\} \\
 & \geq\left|\int_{0}^{t}\d s\left\{ \langle\psi,c^{\ast}(h_{y})\psi\rangle+2\text{Re}\langle h_{y},\eta_{s}\rangle\right\} \right|\nonumber \\
 & \hphantom{\geq}\qquad-\left|\int_{0}^{t}\d s\langle\big(e^{i\widetilde{\mathbb{H}^{\text{eff}}}s}e^{iP(s)}W(\eta_{s})-1\big)^{\ast}\psi,\left(c^{\ast}(h_{y})+2\text{Re}\langle h_{y},\eta_{s}\rangle\right)\psi\rangle\right|\nonumber \\
 & \hphantom{\geq}\qquad\qquad-\int_{0}^{t}\d s\left\{ \|\widetilde{\text{Error}}\|+\|h_{0}W(\eta_{s})\psi\|\right\} \\
 & \geq2\left|\int_{0}^{t}\d s\ \text{Re}\langle h_{y},\eta_{s}\rangle\right|-\big(\|c^{\ast}(h_{y})\psi\|+2\sup_{s\in[0,t]}|\text{Re}\langle h_{y},\eta_{s}\rangle|\big)\int_{0}^{t}\d s\|\big(e^{i\widetilde{\mathbb{H}^{\text{eff}}}s}e^{iP(s)}W(\eta_{s})-1\big)\psi\|\nonumber \\
 & \hphantom{\geq}\qquad\qquad-t\sup_{s\in[0,t]}\left\{ \|\widetilde{\text{Error}}\|+\|h_{0}W(\eta_{s})\psi\|\right\} \label{eq:coherent-lower-bound1}
\end{align}
where we used $\langle\psi,c^{\ast}(h_{y})\psi\rangle=0$ since $c(h_{y})\phi\otimes\Omega=0$. 

The three parts of the above inequality \eqref{eq:coherent-lower-bound1}
can be bounded in the following: 

Firstly, the error term can be bounded with \eqref{eq:h_0-bound}
and analogously to \eqref{eq:Error-estimate} and \eqref{eq:Error2-estimate}:
\begin{align}
 & t \sup_{s\in[0,t]}\left\{ \|\widetilde{\text{Error}}\|+\|h_{0}W(\eta_{s})\psi\|\right\} \nonumber \\
 & \leq C\lambda\kF t MN^{-\frac{2}{3}+\delta}(f_{C}( \lambda\kF t)-1)+C\beta\left(\log\left[f_{1}( \lambda\kF t)\right]+\log\left[f_{1}( \lambda\kF t)\right]^{4}\right)\left(f_{C}( \lambda\kF t)+1\right)\nonumber \\
 & \leq CQ_{V}(\kF t)\max\{\lambda\kF t MN^{-\frac{2}{3}+\delta},\beta t\} \nonumber \\
 & \leq C \max\{MN^{-\frac{2}{3}+\delta},\beta\kF^{-1}\lambda^{-1}\}  \eqqcolon d\label{eq:d-coherent}
\end{align}
where we used $0\leq t \lesssim \mathcal{O}(\kF^{-1}\lambda^{-1})$ and defined the error variable $d$.

Secondly, the prefactor of the integral term is bounded by \prettyref{lem:Pair-operator-bounds}
\begin{align}
 & \|c^{\ast}(h_{y})\psi\|+2\sup_{s\in[0,t]}|\text{Re}\langle h_{y},\eta_{s}\rangle|\nonumber \\
 & \leq\|h_{y}\|\|(\mathcal{N}+1)^{1/2}\psi\|+2\sup_{s\in[0,t]}\sum_{\kinG}\sumak\frac{1-\cos(s\epsilon_{\alpha}(k))}{\epsilon_{\alpha}(k)}|\phiak|^{2}\\
 & \leq\|h_{y}\|+\frac{2}{\kF}\|h_{y}\|^{2}\leq C\lambda\kF\big(\|\hat{V}\|+\lambda\|\hat{V}\|^{2}\big)\eqqcolon \theta\label{eq:alpha-coherent}
\end{align}
where we introduced the variable $\theta$.

Thirdly, we want to show that the remaining term $|\int_{0}^{t}\d s\,\text{Re}\langle h_{y},\eta_{s}\rangle|$
is large: 
\begin{align}
\left|\int_{0}^{t}\d s\,\text{Re}\langle h_{y},\eta_{s}\rangle\right| & =\left|\int_{0}^{t}\d s\ \text{Re}\sum_{\kinG}\sumak\frac{e^{-is\epsilon_{\alpha}(k)}-1}{\epsilon_{\alpha}(k)}|\phiak|^{2}\right|\nonumber \\
 & =\left|\int_{0}^{t}\d s\ \text{Re}\sum_{\kinG}\lambda^{2}\hat{V}(k)^{2}\sumak\frac{e^{-is\epsilon_{\alpha}(k)}-1}{\epsilon_{\alpha}(k)}n_{\alpha}(k)^{2}\right|\nonumber \\
 & =\frac{\pi\lambda^{2}\kF^{2}}{\kF}\left|\int_{0}^{t}\d s\ \text{Re}\sum_{\kinG}\hat{V}(k)^{2}\frac{e^{-i2\kF|k|s}-1}{2\kF|k|s}\right|\left\{ 1+\mathcal{O}\left(M^{\frac{1}{2}}N^{-\frac{1}{3}+\delta}+N^{-\delta}\right)\right\} 
\end{align}
where we used that the sum over $\alpha$ is approximated similarly
to \eqref{lem:Approx-n_alpha(k)^2} by integrating over the half sphere
\begin{align}
\sumak\frac{e^{-is\epsilon_{\alpha}(k)}-1}{\epsilon_{\alpha}(k)}n_{\alpha}(k)^{2} & =2\pi\int_{0}^{\pi/2}\d\theta\ \frac{e^{-is2\kF|k|\cos\theta}-1}{2\kF|k|\cos\theta}\cos\theta\sin\theta\left\{ 1+\mathcal{O}\left(M^{\frac{1}{2}}N^{-\frac{1}{3}+\delta}+N^{-\delta}\right)\right\} \nonumber \\
 & =i\frac{2\pi\kF^{2}}{2\kF}\frac{e^{-i2\kF|k|s}-1}{2\kF|k|s}\left\{ 1+\mathcal{O}\left(M^{\frac{1}{2}}N^{-\frac{1}{3}+\delta}+N^{-\delta}\right)\right\} .
\end{align}
Therefore we obtain 
\begin{equation}
\left|\int_{0}^{t}\d s\ \text{Re}\sum_{\kinG}\hat{V}(k)^{2}\frac{e^{-i2\kF|k|s}-1}{2\kF|k|s}\right|=\left|\sum_{\kinG}\hat{V}(k)^{2}\int_{0}^{t}\d s\frac{1-\cos(2\kF|k|s)}{2\kF|k|s}\right|.
\end{equation}

It is well-known that for all $x>0$ it holds $\cos(x)<1-4x^{2}/\pi^{2}$
and thus $1-\cos(x)\geq4x^{2}/\pi^{2}$ for all $x\in(0,\pi/2)$ .
Also it holds for all $x>\pi/2$ that
\[
\int_{\xo}^{x}\frac{1-\cos y}{y}\d y=\ln(x)-\ln(\xo)+\text{Ci}(\frac{\pi}{2})-\text{Ci}(x)\geq\ln(x)-\ln(\xo)
\]
since $\text{Ci}(\pi/2)>0$ and $\text{Ci}(\pi/2)\geq\text{Ci}(x)$
for all $x\geq\pi/2$. Thus we obtain
\begin{align}
\int_{0}^{x}\frac{1-\cos y}{y}\d y & \geq\chi(x>\xo)\left\{ \int_{0}^{\xo}\frac{4y}{\pi^{2}}\d y+\int_{\xo}^{x}\frac{1-\cos y}{y}\d y\right\} +\chi(x\leq\xo)\int_{0}^{x}\frac{4y}{\pi^{2}}\d y\\
 & =\chi(x>\xo)\left\{ \frac{1}{2}+\ln(x)-\ln(\xo)\right\} +\chi(x\leq\xo)\frac{2x^{2}}{\pi^{2}}
\end{align}
which is a differentiable lower bound.

Therefore we find
\begin{align}
 & 2\left|\int_{0}^{t}\d s\ \text{Re}\langle h_{y},\eta_{s}\rangle\right|\nonumber \\
 & \geq\frac{2\pi\lambda^{2}\kF^{2}}{\kF}\sum_{\kinG}\frac{\hat{V}(k)^{2}}{2\kF|k|}\bigg(\chi(2\kF|k|t>\frac{\pi}{2})\left\{ \frac{1}{2}+\log(2\kF|k|t)-\ln(\frac{\pi}{2})\right\} \nonumber \\
 & \hphantom{\geq\frac{2\pi\lambda^{2}\kF^{2}}{\kF}\sum_{\kinG}\frac{\hat{V}(k)^{2}}{2\kF|k|}\bigg(}+\chi(2\kF|k|t\leq\frac{\pi}{2})\frac{2(2\kF|k|t)^{2}}{\pi^{2}}\bigg)\left\{ 1+\mathcal{O}\left(M^{\frac{1}{2}}N^{-\frac{1}{3}+\delta}+N^{-\delta}\right)\right\} \nonumber \\
 & \eqqcolon b_{t} +d \label{eq:b_t-coherent}
\end{align}
with $d\in\max\{MN^{-\frac{2}{3}+\delta},\beta\kF^{-1}\lambda^{-1}\} $ defined in \eqref{eq:d-coherent}
and 
\begin{align}
\dot{b}_{t} & =\frac{\pi\lambda^{2}\kF^{2}}{\kF}\sum_{\kinG}\hat{V}(k)^{2}\left(\chi(2\kF|k|t>\frac{\pi}{2})\frac{1}{2\kF|k|t}+\chi(2\kF|k|t\leq\frac{\pi}{2})\frac{8\kF|k|t}{\pi^{2}}\right).\label{eq:dot_b}
\end{align}

In total, we can re-write the inequality \eqref{eq:coherent-lower-bound1}
for $g_{t}\coloneqq\|\big(e^{i\widetilde{\mathbb{H}^{\text{eff}}}t}e^{iP(t)}W(\eta_{t})-1\big)\psi\|$
as the following integral inequality
\begin{equation}
g_{t}\geq b_{t}-\theta\int_{0}^{t}\d sg_{s}
\end{equation}
with $b_{t}$ and $\theta$ defined in \eqref{eq:b_t-coherent} and
\eqref{eq:alpha-coherent}, respectively.
\begin{claim*}
We can bound $g_{t}$ from below for all $t\geq0$ by a differentiable
map $h_{t}$ which obeys the initial value problem
\begin{align}
h_{t} & =b_{t}-\theta\int_{0}^{t}\d sh_{s}\qquad\text{with }h_{0}=-d<0=g_{0}.\label{eq:ode-coherent}
\end{align}
\end{claim*}
\begin{proof}[Proof of claim]
Assume for the proof by contradiction that there exists a $ t_{0}>0:h_{t_{0}}>g_{t_{0}}$
and set $I\subseteq\R_{>0}$ as the largest open interval satisfying $t_{0}\in I$
and $h_{t}>g_{t}$ for all $t\in I$. Note that since $h_{t}$ is
differentiable such an open interval exists. It holds $t_{1}\coloneqq\inf I>0$
since $h_{0}<g_{0}$ and for all $t\in I$ 
\[
\dot{g}_{t}\geq\dot{b}_{t}-\theta g_{t}>\dot{b}_{t}-\theta h_{t}=\dot{h}_{t}
\]
 and therefore
\[
\int_{t_{1}}^{t_{0}}\dot{g}_{s}\d s>\int_{t_{1}}^{t_{0}}\dot{h}_{s}\d s\implies h_{t_{1}}-g_{t_{1}}>h_{t_{0}}-g_{t_{0}}>0.
\]
Thus we obtain the desired contradiction to $t_{1}$ being the infimum
of the largest set satisfying $h_{t}>g_{t}$.
\end{proof}
The solution of the initial value problem \eqref{eq:ode-coherent}
is uniquely given by 
\begin{equation}
h_{t}=e^{-\theta t}\int_{0}^{t}\dot{b}_{s}e^{\theta s}\d s -d .
\end{equation}

Consequently it holds by inserting \eqref{eq:dot_b}
\begin{align}
h_{t} & =e^{-\theta t}\int_{0}^{t}\dot{b}_{s}e^{\theta s}\d s\nonumber \\
 & =\frac{\pi\lambda^{2}\kF^{2}}{\kF}\sum_{\kinG}\hat{V}(k)^{2}e^{-\theta t}\int_{0}^{t}\bigg(\chi(2\kF|k|s>\frac{\pi}{2})e^{\theta s}\frac{1}{2\kF|k|s}\nonumber \\
 & \hphantom{=\frac{\pi\lambda^{2}\kF^{2}}{\kF}\sum_{\kinG}\hat{V}(k)^{2}e^{-\theta t}\int_{0}^{t}\bigg(}+\chi(2\kF|k|s\leq\frac{\pi}{2})e^{\theta s}\frac{8\kF|k|s}{\pi^{2}}\bigg)\d s -d\\
 & \geq\pi\lambda^{2}\kF\sum_{\kinG}\hat{V}(k)^{2}\bigg(\chi(|k|t\leq\frac{\pi}{4\kF})\frac{8\kF|k|}{\theta^{2}\pi^{2}}\left(e^{-\theta t}+\theta t-1\right)+\nonumber \\
 & \hphantom{\geq\frac{\pi\lambda^{2}\kF}{}\sum_{\kinG}\hat{V}(k)^{2}}+\chi(|k|t>\frac{\pi}{4\kF})\frac{8\kF|k|}{\theta^{2}\pi^{2}}\big(e^{-\theta\pi/(4\kF|k|)}+\frac{\theta\pi}{4\kF|k|}-1\big)\bigg) -d\\
 & \geq\frac{\lambda^{2}\kF^{2}}{\theta^{2}\pi}\sum_{\kinG}\hat{V}(k)^{2}|k|\min\bigg( f(\theta t),f\big(\frac{\pi \theta }{4\kF|k|}\big)\bigg) -d
\end{align}
with $f(t)\coloneqq\left(e^{- t}+ t-1\right)$ defining
a non-negative monotonically increasing function. 
Recall that $d\in\max\{MN^{-\frac{2}{3}+\delta},\beta\kF^{-1}\lambda^{-1}\} $ from \eqref{eq:d-coherent} 
for all $0\leq t \lesssim \kF^{-1} \lambda^{-1}$
and $\theta \leq C\lambda\kF$ with a constant $C>0$ depending only on $V$ from \prettyref{eq:alpha-coherent}.
Thus, we obtain the desired result that \eqref{eq:coherent-state-error-II}
has a lower bound of order 1.
\end{proof}

\subsubsection*{}

\newpage{}

\appendix

\section{\label{sec:Useful-bosonization}Estimates within the bosonization
framework}

In this section we collect all relevant estimates of the bosonization
framework which was first developed in \cite{Benedikter2019,Benedikter2021,Benedikter2021a,Benedikter2023}
in the semiclassical regime. In addition to the references we present
brief proof sketches where we think they are helpful for the interested reader.

\begin{lem}[{Approximation of $n_{\alpha}(k)$, \cite[Lemma 5.1]{Benedikter2021a}}]
\label{lem:Approx-n_alpha(k)^2-wo-sum}For $N^{2\delta}\ll M\ll N^{\frac{2}{3}-2\delta}$
and $\kinG,\alpha\in\mathcal{I}_{k}$ it holds 
\[
n_{\alpha}(k)^{2}=\frac{4\pi\kF^{2}}{M}|k\cdot\hat{\omega}_{\alpha}|\big(1+o(1)\big).
\]
Note that $|k\cdot\hat{\omega}_{\alpha}|>N^{-\delta}$ by construction
of $\mathcal{I}_{k}$.
\end{lem}

\begin{lem}[Approximation of $\sumak n_{\alpha}(k)^{2}$]
\label{lem:Approx-n_alpha(k)^2} For $N^{2\delta}\ll M\ll N^{\frac{2}{3}-2\delta}$
and $\kinG$ it holds 
\[
\sumak n_{\alpha}(k)^{2}=\kF^{2}|k|\pi\left\{ 1+\mathcal{O}\left(M^{\frac{1}{2}}N^{-\frac{1}{3}+\delta}+N^{-\delta}\right)\right\} .
\]
\end{lem}

\begin{proof}
It holds from \cite[Proposition 3.1]{Benedikter2019} 
\[
n_{\alpha}(k)^{2}=\kF^{2}|k|\sigma(p_{\alpha})u_{\alpha}(k)^{2}\left(1+\mathcal{O}(M^{\frac{1}{2}}N^{-\frac{1}{3}+\delta})\right)
\]
with $\cos\theta_{\alpha}\coloneqq|\hat{k}\cdot\hat{\omega}_{\alpha}|\equiv u_{\alpha}(k)^{2}$.
Choose $\varphi_{\alpha}$ for the azimuth angle of $\omega_{\alpha}$.
We estimate the $\alpha$-sum by an appropriate surface integral over
the patch $p_{\alpha}$
\begin{align*}
\left|\int_{p_{\alpha}}\d\sigma\ \cos\theta-\sigma(p_{\alpha})\cos\theta_{\alpha}\right| & \leq\sup_{\hat{\omega}(\theta,\varphi)\in p_{\alpha}}\left|\frac{\d}{\d\theta}\cos\theta\right|\sup_{(\theta,\varphi)\in p_{\alpha}}\left|\theta-\theta_{\alpha}\right|\sigma(p_{\alpha})\\
 & \leq CM^{-\frac{3}{2}}
\end{align*}
where we used $\left|\theta-\theta_{\alpha}\right|\leq CM^{-\frac{1}{2}}$
and $\sigma(p_{\alpha})\leq CM^{-1}$ by the patch construction. Note
that the integral over $\tilde{S}\coloneqq\bigcup_{\alpha\in\mathcal{I}_{k}}p_{\alpha}$
which excludes a collar of width $N^{-\delta}$ can be approximated
by an integral over the half-sphere $S$ 
\[
\left|\int_{\tilde{S}}\d\sigma\ \cos\theta-\int_{S}\d\sigma\ \cos\theta\right|<C\left(N^{-\delta}+M^{\frac{1}{2}}N^{-\frac{1}{3}}\right)
\]
which can be calculated explicitly
\[
\int_{S}\d\sigma\ \cos\theta=2\pi\int_{0}^{\pi/2}\d\theta\ \cos\theta\sin\theta=\pi.
\]
Thus in total we obtain with the triangle inequality and $|k|<C$
\begin{align*}
\left|\sumak n_{\alpha}(k)^{2}-\kF^{2}|k|\pi\right| & =\left|\sumak n_{\alpha}(k)^{2}-\kF^{2}|k|\int_{S}\d\sigma\ \cos\theta\right|\\
 & \leq C\kF^{2}|k|\left(M^{\frac{1}{2}}N^{-\frac{1}{3}+\delta}+M^{-\frac{1}{2}}+N^{-\delta}+M^{\frac{1}{2}}N^{-\frac{1}{3}}\right)\\
 & \leq C\kF^{2}|k|\left(M^{\frac{1}{2}}N^{-\frac{1}{3}+\delta}+N^{-\delta}\right).
\end{align*}
\end{proof}
\begin{lem}[{Estimates on the CCR error term, \cite[Lemma 5.2]{Benedikter2021a}}]
\label{lem:CAR-error} For $k',\kinG$ and $\alpha\in\mathcal{I}_{k},\beta\in\mathcal{I}_{k'}$
the error term $\mathcal{E}_{\alpha}(k,k')$ as defined in \eqref{eq:approx-CAR}
satisfies $\mathcal{E}_{\alpha}(k,k')=\mathcal{E}_{\alpha}(k',k)^{\ast}$,
commutes with $\mathcal{N}$ and for all $\gamma\in\mathcal{I}_{k}\cap\mathcal{I}_{k'}$
it holds for all $\zeta\in\mathcal{F}$ 
\begin{align*}
|\mathcal{E}_{\gamma}(k,k')|^{2}\leq\sum_{\gamma\in\mathcal{I}_{k}\cap\mathcal{I}_{k'}}|\mathcal{E}_{\gamma}(k,k')|^{2} & \leq C\left(MN^{-\frac{2}{3}+\delta}\mathcal{N}\right)^{2},\\
\sum_{\gamma\in\mathcal{I}_{k}\cap\mathcal{I}_{k'}}\|\mathcal{E}_{\gamma}(k,k')\zeta\| & \leq CM^{\frac{3}{2}}N^{-\frac{2}{3}+\delta}\|\mathcal{N}\zeta\|.
\end{align*}
\end{lem}

\begin{lem}[{Pair operator bounds, \cite[Lemma 5.3]{Benedikter2021a}}]
\label{lem:Pair-operator-bounds}It holds for all $\kinG$ and $\psi\in\mathcal{F},f\in l^{2}(\mathcal{I}_{k})$: 
\end{lem}

\begin{enumerate}
\item $\sumak\|\cak\psi\|^{2}\leq\|\mathcal{N}^{\frac{1}{2}}\psi\|^{2}$,
\item $\sumak\|\cak\psi\|\leq M^{\frac{1}{2}}\|\mathcal{N}^{\frac{1}{2}}\psi\|$,
\item $\sumak\|\ccak\psi\|\leq M^{\frac{1}{2}}\|(\mathcal{N}+M)^{\frac{1}{2}}\psi\|$,
\item $\sumak\|\ccak\psi\|^{2}\leq\|(\mathcal{N}+M)^{\frac{1}{2}}\psi\|^{2}$,
\item $\|\sumak f_{\alpha}\ccak\psi\|\leq\|f\|_{l^{2}}\|(\mathcal{N}+1)^{\frac{1}{2}}\psi\|$,
\item $\sumak\ccak\cak\leq\mathcal{N}$.
\end{enumerate}
\begin{lem}[{Error of linearized kinetic energy, \cite[Lemma 8.2]{Benedikter2021a}}]
\label{lem:Elin-bound}It holds for $\kinG$, $\alpha\in\mathcal{I}_{k}$
and all $\psi\in\fock$ 
\[
[\mathbb{H}_{0},\ccak]=2\kF|k\cdot\hat{\omega}_{\alpha}|\ccak+\Elin_{\alpha}(k)^{\ast}
\]
with
\begin{align*}
\sumak\|\Elin_{\alpha}(k)\psi\|^{2} & \leq C\left( N^{\frac{1}{3}}M^{-\frac{1}{2}}\right)^{2}\|(\mathcal{N}+1)^{\frac{1}{2}}\psi\|^{2}\\
\sumak\|\Elin_{\alpha}(k)\psi\| & \leq C N^{\frac{1}{3}}\|\mathcal{N}^{\frac{1}{2}}\psi\|.
\end{align*}
\end{lem}

\begin{proof}
Observe that
\begin{align*}
[\mathbb{H}_{0},\ccak] & =\frac{1}{n_{\alpha}(k)}\sum_{p\in\BF^{c}\cap B_{\alpha},p-k\in\BF\cap B_{\alpha}}\sum_{l\in\mathbb{Z}^{3}}\Big[e(l)a_{l}^{\ast}a_{l},a_{p}^{\ast}a_{p-k}^{\ast}\Big]\\
 & =\frac{1}{n_{\alpha}(k)}\sum_{p\in\BF^{c}\cap B_{\alpha},p-k\in\BF\cap B_{\alpha}}\left(e(p)+e(p-k)\right)a_{p}^{\ast}a_{p-k}^{\ast}\\
 & =2\kF|k\cdot\hat{\omega}_{\alpha}|\ccak+\Elin_{\alpha}(k)^{\ast}
\end{align*}
One makes the identification $\Elin_{\alpha}(k)\equiv c_{\alpha}^{g}(k)$
representing a weighted operator (see \cite[eq. (5.11)]{Benedikter2021a}) with 
\begin{align*}
g(p,k) & =e(p)+e(p-k)-2\kF|k\cdot\hat{\omega}_{\alpha}|=|p|^{2}-|p-k|^{2}-2\kF|k\cdot\hat{\omega}_{\alpha}|\\
 & =\left(2k\cdot(p-\kF\hat{\omega}_{\alpha})-|k|^{2}\right)
\end{align*}
where we used $e(p)$ as defined in \eqref{eq:H_0-def}. The bound
follows from \cite[Lemma 5.4]{Benedikter2021a} which depends on $\|g\|_{l^{\infty}}$.
This can be estimated by using $\text{diam}(B_{\alpha})\leq CN^{\frac{1}{3}}M^{-\frac{1}{2}}$
such that $|g(p,k)|\leq C N^{\frac{1}{3}}M^{-\frac{1}{2}}$.
\end{proof}
\begin{lem}[{Error of bosonized kinetic energy, \cite[eq. (8.6)]{Benedikter2021a}}]
\label{lem:Ebos-bound}It holds for $\kinG$, $\alpha\in\mathcal{I}_{k}$
and all $\psi\in\fock$ 
\[
[\mathbb{D}_{\textnormal{B}},\ccak]=2\kF|k\cdot\hat{\omega}_{\alpha}|\ccak+\Ebos_{\alpha}(k)^{\ast}
\]
with
\begin{align*}
\sumak\|\mathfrak{E}_{\alpha}^{\text{B}}(k)\psi\|^{2} & \leq C\left(\kF MN^{-\frac{2}{3}+\delta}\right)^{2}\|(\mathcal{N}+1)^{\frac{3}{2}}\psi\|^{2},\\
\sumak\|\mathfrak{E}_{\alpha}^{\text{B}}(k)\psi\| & \leq C\kF M^{\frac{3}{2}}N^{-\frac{2}{3}+\delta}\|(\mathcal{N}+1)^{\frac{3}{2}}\psi\|.
\end{align*}
\end{lem}

\begin{proof}
It holds $\mathfrak{E}_{\alpha}^{\text{B}}(k)=2\kF\sum_{l}|l\cdot\hat{\omega}_{\alpha}|\mathcal{E}_{\alpha}^{\ast}(l,k)c_{\alpha}(l)\chi(\alpha\in\mathcal{I}_{l})$
and therefore
\begin{align*}
 & \sumak\|\mathfrak{E}_{\alpha}^{\text{B}}(k)\psi\|^{2}\\
 & \leq\sum_{\alpha\in\mathcal{I}_{k}\cap\mathcal{I}_{l}}\bigg(\sum_{\linG}||2\kF|l\cdot\hat{\omega}_{\alpha}|\mathcal{E}_{\alpha}^{\ast}(l,k)c_{\alpha}(l)\psi\|\bigg)^{2}\leq C\kF\sum_{\alpha\in\mathcal{I}_{k}\cap\mathcal{I}_{l}}\bigg(\sum_{\linG}||\mathcal{E}_{\alpha}^{\ast}(l,k)c_{\alpha}(l)\psi\|\bigg)^{2}\\
 & \leq C\kF\sum_{\alpha\in\mathcal{I}_{k}\cap\mathcal{I}_{l}}\bigg(\sum_{\linG}CMN^{-\frac{2}{3}+\delta}||\mathcal{N}c_{\alpha}(l)\psi\|\bigg)^{2}\\
 & \leq C\kF\left(CMN^{-\frac{2}{3}+\delta}\right)^{2}\sum_{\alpha\in\mathcal{I}_{k}\cap\mathcal{I}_{l}}\sum_{\linG}\sum_{l'\in\Gamma}\|c_{\alpha}(l)(\mathcal{N}-2)\psi\|\ \|c_{\alpha}(l')(\mathcal{N}-2)\psi\|\\
 & \leq C\kF\left(CMN^{-\frac{2}{3}+\delta}\right)^{2}\|(\mathcal{N}+1)^{\frac{3}{2}}\psi\|^{2}
\end{align*}
where we used in the second line $\linG$ bounded, \prettyref{lem:CAR-error}
in the third line, $\mathcal{N}c_{\alpha}(l)=c_{\alpha}(l)(\mathcal{N}-2)$
in the fourth line and Cauchy-Schwarz and \prettyref{lem:Pair-operator-bounds}
in the last line.

The second statement simply follows from Cauchy-Schwarz.
\end{proof}
\begin{lem}[Approximation of patch decomposed operators]
\label{lem:Approx-of-patch-operators}It holds for all $\kinG$

\[
\|\big(b(k)-\sumak n_{\alpha}(k)\cak+\text{h.c.}\big)\psi\|\leq C(N^{\frac{1}{3}-\frac{\delta}{2}}+N^{\frac{1}{6}}M^{\frac{1}{4}})\|(\mathcal{N}+1)^{\frac{1}{2}}\psi\|.
\]
\end{lem}

\begin{lem}[{Estimate of non-bosonizable terms, \cite[eq. (4.6)]{Benedikter2021}}]
\label{lem:Estimate-non-bosonizable} It holds for $\mathcal{E}$
as defined in \eqref{eq:non-bosonizable-V} the following estimate
for all $\psi\in\fock$

\[
\|\mathcal{E}\psi\|\leq C \lambda\|\hat{V}\|_{1}\|\mathcal{N}\psi\|.
\]
\end{lem}

\section*{Acknowledgment}

We are grateful for helpful discussions with David Mitrouskas,
Benjamin Schlein and Karla Schön. This work has been partially supported by SFB/Transregio
TRR 352 -- Project-ID 470903074.

\subsection*{Declaration}
\paragraph*{Conflict of interest} The authors declare that there is no conflict of interest.

\bibliographystyle{alphadin}
\bibliography{Promotion_DB}

\end{document}